\documentclass[12pt,draftclsnofoot,onecolumn]{IEEEtran}
\usepackage{amsthm,amssymb,graphicx,multirow,amsmath,color,amsfonts}

\usepackage{subcaption}
\usepackage{caption}
\usepackage[update,prepend]{epstopdf}
\usepackage[noadjust]{cite}
\usepackage[latin1]{inputenc}
\usepackage{tikz}
\usepackage{bbm} 
\usepackage{pdfpages}
\usepackage{multirow}
\usepackage{comment}

\usepackage{algorithm}
\usepackage{algpseudocode}
\usepackage{url}
\def\BSTATE{\STATE\hskip-\ALG@thistlm}
\makeatother

\def\nb0{{\mathbf{0}}}
\def\nb1{{\mathbf{1}}}

\newtheorem{lemma}{Lemma}

\newtheorem{theorem}{Theorem}
\newtheorem{prop}{Proposition}

\newtheorem{remark}{Remark}

\def\argmin{\operatorname{arg~min}}
\def\argmax{\operatorname{arg~max}}
\allowdisplaybreaks 

\usepackage{setspace}	

\makeatletter
\newcommand\semihuge{\@setfontsize\semihuge{22.3}{22}}
\makeatother

\begin{document}
\graphicspath{{./Figures/}}
\title{\semihuge{
Neural Combinatorial Deep Reinforcement Learning for Age-optimal Joint Trajectory and Scheduling Design in UAV-assisted Networks}}
\author{
Aidin Ferdowsi, Mohamed A. Abd-Elmagid, Walid Saad, and Harpreet S. Dhillon \vspace{-18mm}
\thanks{The authors are with Wireless@VT, Department of ECE, Virginia Tech, Blacksburg, VA (Emails: \{aidin,maelaziz,walids,hdhillon\}@vt.edu). This research was supported in part by the U.S. National Science Foundation under Grant CNS-1814477 and in part by the the Office of Naval Research (ONR) under MURI Grant N00014-19-1-2621. 
}
}

\maketitle

\begin{abstract}
In this paper, an unmanned aerial vehicle (UAV)-assisted wireless network is considered in which a battery-constrained UAV is assumed to move towards energy-constrained ground nodes to receive status updates about their observed processes. The UAV's flight trajectory and scheduling of status updates are jointly optimized with the objective of minimizing the normalized weighted sum of Age of Information (NWAoI) values for different physical processes at the UAV. The problem is first formulated as a mixed-integer program. Then, for a given scheduling policy, a convex optimization-based solution is proposed to derive the UAV's optimal flight trajectory and time instants on updates. However, finding the optimal scheduling policy is challenging due to the combinatorial nature of the formulated problem. Therefore, to complement the proposed convex optimization-based solution, a finite-horizon Markov decision process (MDP) is used to find the optimal scheduling policy. Since the state space of the MDP is extremely large, a novel neural combinatorial-based deep reinforcement learning (NCRL) algorithm using deep Q-network (DQN) is proposed to obtain the optimal policy. However, for large-scale scenarios with numerous nodes, the DQN architecture cannot efficiently learn the optimal scheduling policy anymore. Motivated by this, a long short-term memory (LSTM)-based autoencoder is proposed to map the state space to a fixed-size vector representation in such large-scale scenarios while capturing the spatio-temporal interdependence between the update locations and time instants. A lower bound on the minimum NWAoI is analytically derived which provides system design guidelines on the appropriate choice of importance weights for different nodes. Furthermore, an upper bound on the UAV's minimum speed is obtained to achieve this lower bound value. The numerical results also demonstrate that the proposed NCRL approach can significantly improve the achievable NWAoI per process compared to the baseline policies, such as weight-based and discretized state DQN policies.
\end{abstract}\vspace{-5mm}
\begin{IEEEkeywords}\vspace{-5mm}
\noindent Age of information, unmanned aerial vehicles, deep reinforcement learning, convex optimization.
\end{IEEEkeywords}
\IEEEpeerreviewmaketitle
\section{Introduction} \label{sec:intro}
Owing to their flexible deployment, the unmanned aerial vehicles (UAVs) have emerged as a key component of future wireless networks. The use of UAVs as flying base stations (BSs), that collect/transmit information from/to ground nodes (e.g., users, sensors or Internet of Things (IoT) devices), has recently attracted significant attention \cite{mozaffari2019tutorial,challita2019machine,bor2016efficient,azari2016joint,alzenad20173,mozaffari2018beyond,Eldosouky,kishk20203}. Meanwhile, introducing UAVs into wireless networks leads to many challenging design questions related to optimal deployment, flight trajectory design, and energy efficiency, to name a few. So far, these challenges have mostly been addressed in the literature using traditional performance metrics such as network coverage, rate and delay. However, such performance metrics lack the ability of quantifying the freshness of information collected by the UAVs since they do not account for the generation times of the information at the ground nodes. 
As a result, these existing solutions are not always suitable for many real-time monitoring applications, such as safety and IoT applications, whose quality-of-service (QoS) depends upon the freshness of the collected information when it reaches the UAV \cite{abd2018role}. This necessitates the design of new freshness-aware transmission policies that can efficiently guide the UAV's flight trajectory as well as carefully schedule information transmissions from the ground nodes, which is the main objective of this work.
\subsection{Related works}
Trajectory planning for UAVs has gained considerable attention in the recent past \cite{zeng2016throughput,xie2018throughput,li2018placement,Msamir,farajzadeh2019uav,banagar2020performance,monwar2018optimized,Ydu,Fcui,chen2020trajectory,hu2020meta}. 
The works in \cite{zeng2016throughput,xie2018throughput,li2018placement,Msamir} formulated a non-convex optimization problem to derive an optimal trajectory of the UAV that maximizes the total throughput of the network while taking into consideration the energy limitations of the UAV and ground nodes. Then, different successive convex optimization solutions were proposed to reduce the complexity of the problem. The authors in \cite{farajzadeh2019uav} jointly optimized the UAV's flight trajectory and altitude with the objective of maximizing the total throughput of UAV-assisted backscatter networks. Using tools from stochastic geometry, the authors in \cite{banagar2020performance} characterized the performance of several canonical mobility models in an UAV cellular network. Meanwhile, heuristic methods, flow-shop scheduling, dual decomposition, shortest path, and meta reinforcement learning (RL) techniques have been proposed in \cite{monwar2018optimized,Ydu,Fcui,chen2020trajectory,hu2020meta} for energy efficient and maximal throughput trajectory design in UAV-assisted wireless networks. However, the flight trajectories considered in \cite{zeng2016throughput,xie2018throughput,li2018placement,Msamir,farajzadeh2019uav,banagar2020performance,monwar2018optimized,Ydu,Fcui,chen2020trajectory,hu2020meta} may not necessarily be optimal from the perspective of preserving freshness of the status updates since they were obtained using traditional performance metrics, such as throughput and delay. 

We adopt the concept of age of information (AoI) to quantify the freshness of information at the UAV. First introduced in \cite{kaul2012real}, AoI is defined as the time elapsed since the latest received status update packet at a destination node was generated at the source node. For a simple queueing-theoretic model, the work in \cite{kaul2012real} characterized the average AoI, and demonstrated that the optimal rate at which the source should generate its update packets in order to minimize the average AoI is different from the optimal rates that either maximize throughput or minimize delay. Then, the average AoI and other age-related metrics were investigated in the literature for variations of the model considered in \cite{kaul2012real} (see \cite{kosta2017age_mono} for a comprehensive survey). These early works have inspired the adoption of AoI 
as a performance metric for different communication systems that deal with time critical information \cite{sun2017update,kadota2016,hsu2019scheduling,Buyukates_ulu,li2020age_a,zhou2018joint,AbdElmagid2019Globecom_a,Stamatakis_2020,abd2019tcom,zhou2019minimum,wang2020minimizing,AbdElmagid_joint,emara2019spatiotemporal,Praful_GC1,mankar2020stochastic_GC2,abdel2018ultra,bastopcu2019minimizing,altman2019forever}. In particular, AoI has been studied in the context of broadcast networks (e.g., \cite{kadota2016} and \cite{hsu2019scheduling}), multicast networks (\cite{Buyukates_ulu} and \cite{li2020age_a}), transmission scheduling policies \cite{zhou2018joint,AbdElmagid2019Globecom_a,Stamatakis_2020,abd2019tcom,zhou2019minimum,wang2020minimizing,AbdElmagid_joint} and large-scale analysis \cite{emara2019spatiotemporal,Praful_GC1,mankar2020stochastic_GC2} of IoT networks, ultra-reliable low-latency vehicular networks \cite{abdel2018ultra}, and social networks (\cite{bastopcu2019minimizing} and \cite{altman2019forever}).
Note that the prior art in \cite{kaul2012real,kosta2017age_mono,sun2017update,kadota2016,hsu2019scheduling,Buyukates_ulu,li2020age_a,zhou2018joint,AbdElmagid2019Globecom_a,Stamatakis_2020,abd2019tcom,zhou2019minimum,wang2020minimizing,AbdElmagid_joint,emara2019spatiotemporal,Praful_GC1,mankar2020stochastic_GC2,abdel2018ultra,bastopcu2019minimizing,altman2019forever} assumed the destination node to be 
static, and, thus, their results cannot be generalized to a scenario in which the destination is a mobile node such as a UAV.

The use of UAVs for maintaining freshness of information (quantified using AoI) collected from a set of ground nodes has been recently studied in  \cite{UAV_AoI_noRL1,UAV_AoI_noRL2,UAV_AoI_noRL3,UAV_AoI_noRL4,UAV_AoI_RL,UAV_AoI_DRL1,UAV_AoI_DRL2,UAV_AoI_DRL3,UAV_AoI_DRL4}. The authors in \cite{UAV_AoI_noRL1} investigated the role of a UAV as a mobile relay to minimize the average Peak AoI for a source-destination pair model by jointly optimizing the UAV's flight trajectory as well as energy and service time allocations for the transmission of status updates. Dynamic programming-based approaches were proposed in \cite{UAV_AoI_noRL2,UAV_AoI_noRL3} to optimize the UAV's flight trajectory with the objective of minimizing the average of the AoI values associated with different ground nodes. Furthermore, a graph labeling-based algorithm  was developed in \cite{UAV_AoI_noRL4} to determine the optimal scheduling of update transmissions from the ground nodes while assuming that the UAV is equipped with a battery of finite capacity (which needs to be recharged over time).
The works in \cite{UAV_AoI_RL,UAV_AoI_DRL1,UAV_AoI_DRL2,UAV_AoI_DRL3,UAV_AoI_DRL4} proposed techniques from reinforcement learning (RL) to learn age-optimal transmission policies. In particular, in \cite{UAV_AoI_RL}, the authors proposed to use Q-learning for scheduling update transmissions from ground nodes with the objective of minimizing the expired data packets. Meanwhile, deep Q-network (DQN) approaches with different settings were proposed in our early work \cite{UAV_AoI_DRL1} and in \cite{UAV_AoI_DRL2,UAV_AoI_DRL3,UAV_AoI_DRL4} to find an optimal trajectory and/or scheduling policy for the UAV in order to minimize the AoI of ground nodes. However, these works considered discretized trajectory and time instants in their underlying system settings, which introduces approximation errors to the obtained age-optimal policies and limits their implementation in real-world scenarios.

\subsection{Contributions}
The main contribution of this paper is a novel approach that combines tools from convex optimization and deep RL framework for optimizing the UAV's flight trajectory as well as the scheduling of the status update packets from ground nodes with the objective of minimizing the normalized weighted sum of Age of Information (NWAoI) values at the UAV. In particular, we study a UAV-assisted wireless network, in which a UAV moves towards the ground nodes to collect status update packets about their observed processes. For this system setup, we formulate an NWAoI minimization problem in which the UAV's flight trajectory as well as scheduling of update packet transmissions are jointly optimized. 
The problem is solved in two steps. First, a convex optimization-based approach is proposed to derive the trajectory as well as the update time instants of nodes for a specific scheduling policy. Next, in order to find the optimal scheduling policy, a finite-horizon Markov decision process (MDP) model with finite state and action spaces is proposed. Due to the combinatorial nature of the problem of finding the optimal scheduling policy, the use of a finite-horizon dynamic programming (DP) algorithm is computationally impractical. To overcome this challenge, we propose a {\it neural combinatorial RL (NCRL) algorithm} for this setting \cite{bello2016neural} and \cite{khalil2017learning}. Unlike conventional RL problems, we show that the state of our problem has a two dimensional matrix form with varying number of columns. Therefore, we propose a long short-term memory (LSTM)-based autoencoder that can map the state of the problem with varying sizes into a fixed size state representation. 

Several key system design insights are drawn from our analysis. For instance, we analytically derive a lower bound on the minimum NWAoI, which is useful in deciding the importance weights for different nodes. In particular, a key observation from the analytical expression of the lower bound is that in order to have a similar impact from each node on the NWAoI, the importance weight of each node should be chosen such that it is proportional to the total number of updates transmitted by that node. Furthermore, we derive an upper bound on the UAV's minimum speed to achieve this lower bound value. Our numerical results also demonstrate the superiority of the proposed NCRL approach over the baseline policies, such as weight-based and discretized state policies, in terms of the achievable NWAoI per process. They also reveal that the NWAoI monotonically decreases with the battery sizes of the ground nodes, and the UAV's speed and time constraint, whereas it monotonically increases with the number of nodes. 

To the best of our knowledge, this work is the first to combine tools from convex optimization and deep RL to characterize the age-optimal policy in a practical scenario involving a continuous flight trajectory model for the UAV.
\section{System Model and Problem Formulation}\label{sec:Model}
\subsection{Network Model}
Consider a wireless network in which a set $\mathcal{M}$ of $M$ ground nodes are deployed to observe potentially different physical processes (e.g., agricultural, healthcare, safety, or industrial data) of a certain geographical region. Uplink transmissions are considered, where a UAV collects status update packets from the ground nodes while seeking to maintain freshness of its information status about their observed processes during the time of its operation. We assume that each ground node $m \in \mathcal{M}$ has a battery with finite capacity of $E_{m}^{\rm max}$ and its battery level at time instant $ t $ is denoted by $e_{m}(t) \in [0,E_{m}^{\rm max}]$. As shown in Fig. \ref{fig:system model}, the UAV flies at a fixed height $h$ such that the projection of its flight trajectory on the ground at time instant $t$ is denoted by $L_{u}(t) \triangleq (x_u(t),y_u(t)) $, where $ x_u(t) $ and $ y_u(t) $ represent the projection of the UAV's location on the $ x $ and $ y $ axes, respectively. Furthermore, we define $ v_{u,x}(t) $ and $ v_{u,y}(t) $ as the UAV's velocity in the $ x $ and $ y $ directions at time instant $ t $ such that we have:
\begin{align}
	&\frac{{\rm d}x_u(t)}{{\rm d}t} = v_{u,x}(t),\,\, \frac{{\rm d}y_u(t)}{{\rm d}t} = v_{u,y}(t),\label{eq:loc_speed}\\
	&-v^{\textrm{max}}_{x}\leq v_{u,x}(t)\leq v^{\textrm{max}}_{x},\,\, -v^{\textrm{max}}_{y}\leq v_{u,y}(t)\leq v^{\textrm{max}}_{y},\label{eq:max_speed}
\end{align}
where $ v^{\textrm{max}}_{x} $ and $ v^{\textrm{max}}_{y} $ represent the maximum speed of the UAV in the horizontal and vertical directions, respectively. Due to battery constraints, the UAV can only operate for a finite time interval. We model this fact by having a time constraint of $\tau$ seconds during which the UAV flies from an initial location $L_{u}^{\rm i}$ to a final location $L_{u}^{\rm f}$ where it can be recharged to continue its operation. Similar to \cite{zeng2016throughput,li2018placement,xie2018throughput}, the channels between the UAV and ground nodes are assumed to be dominated by the line-of-sight (LoS) links. Therefore, at time instant $t$, the channel power gain between the UAV and ground node $m$ is modeled as:
\begin{align}\label{eq:channel}
g_{u,m}(t) = \beta_0 d_{u,m}^{- 2}(t) = \frac{\beta_0}{h^2 + \lVert L_{u}(t) - L_{m}\rVert^2},\; m \in \mathcal{M},
\end{align} 
where $d_{u,m}(t)$ is the distance between the UAV and node $m$ at time instant $t$, $L_{m} =[x_m,y_m]$ is the location of node $m$, and $\beta_0$ is the channel gain at a reference distance of 1 meter.

\begin{figure}[t!]
    \centering
    \vspace{-3mm}
    \includegraphics[width=0.55\columnwidth]{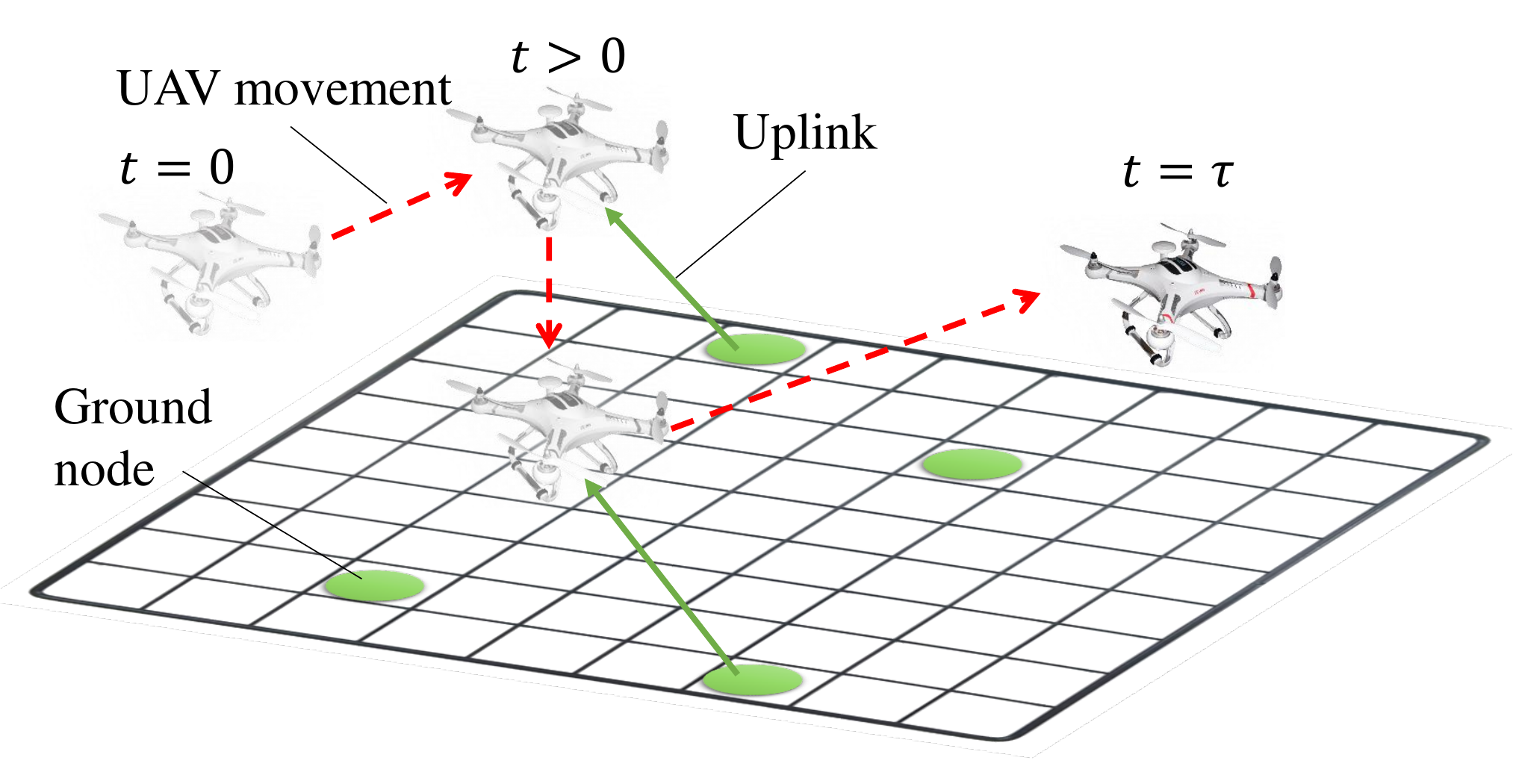}
    \caption{An illustration of our system model.}
    \label{fig:system model}
\end{figure}
 The AoI of an arbitrary physical process is defined as the time elapsed since the most recently received update packet at the UAV was generated at the ground node observing this process. We let $A_{m}(t) \geq A^{\textrm{min}}_{m}$ be the AoI at the UAV for the process observed by node $m$ at time instant $t$, where $A^{\textrm{min}}_{m}$ is the minimum value for $A_{m}(t)$, which is non-zero because of the transmission delay of the wireless link. Since we do not explicitly model this delay in our setup, we simply interpret $A_m^{\rm min}$ as a constant that will correspond to the worst-case transmission delay. Note that this is a reasonable assumption since the value of $A^{\textrm{min}}_{m}$ is negligible compared to the difference between any two consecutive update time instants ($A^{\textrm{min}}_{m}$ is in the order of milliseconds whereas the difference between any two consecutive update time instants is in the order of seconds). Let $ t_{i,m} $ be the time instant at which node $ m $ transmits an update packet for the $ i $-th time. Hence, the AoI dynamics for the process observed by node $m$ will be:
 \begin{align}\label{eq:AoI_evol}
 A_{m}(t) = A^{\textrm{min}}_{m} + t-t_{i-1,m}, \forall t \in[t_{i-1,m},t_{i,m})\,\, \& \,\, i\in \left\{1,\dots,n_m\right\},
 \end{align}
 where $ t_{0,m} \triangleq 0 $ and $ n_m $ is the total number of updates transmitted by node  $ m $. Therefore, as shown in Fig. \ref{fig:AoI}, when $t = t_{i,m}$, the AoI of the observed process is reset to $A_m^{\min}$; otherwise, the AoI value increases linearly. 
 \begin{figure}[t!]
 	\centering
 	\includegraphics[width=0.45\columnwidth]{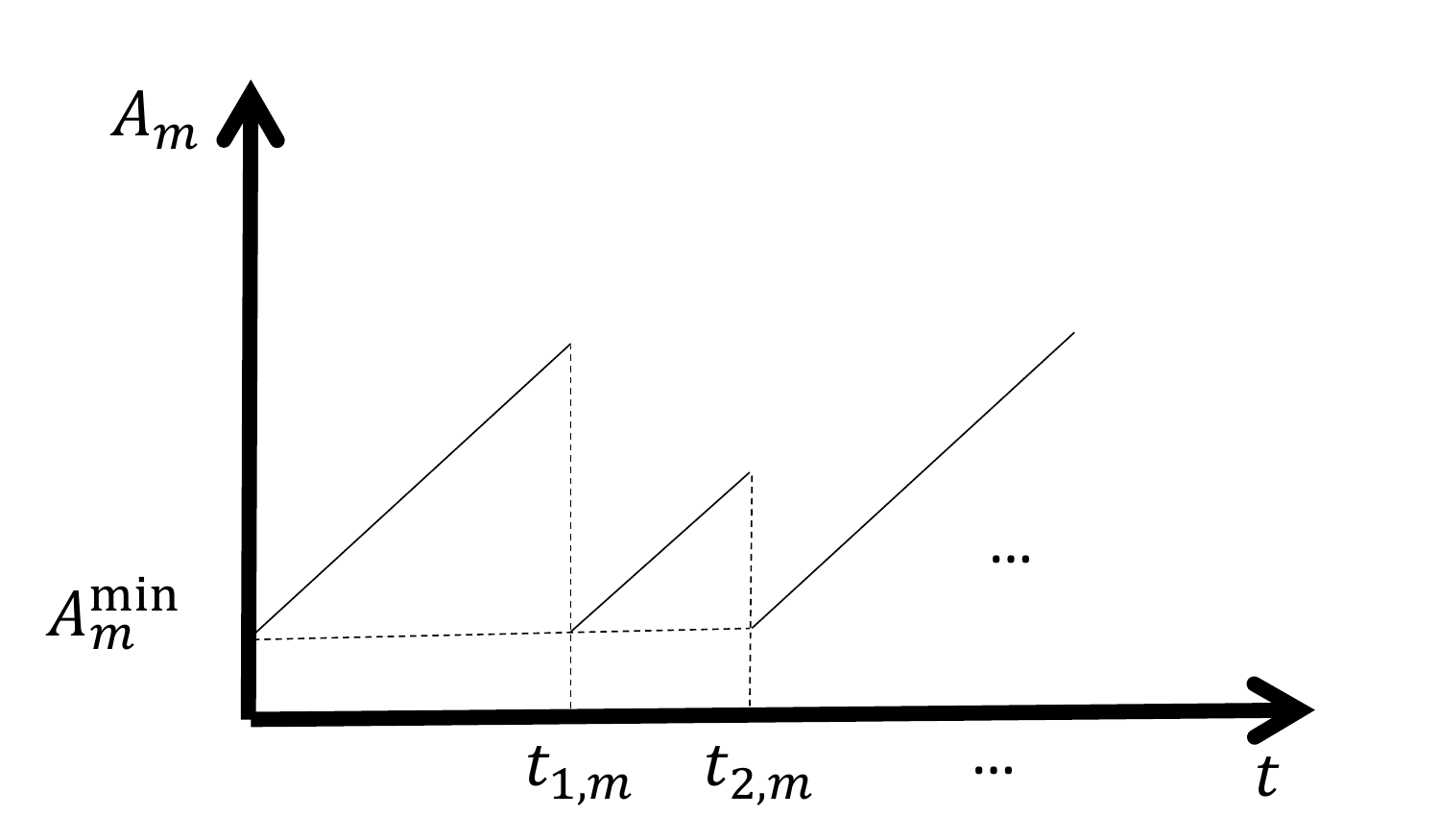}
 	\vspace{-3mm}
 	\caption{AoI evolution vs. update time instants.}
 	\label{fig:AoI}
 \end{figure}
 By letting $S, B$, and $\sigma^2$ be the size of an update packet, channel bandwidth, and noise power at the UAV, respectively, the energy required to transmit an update packet from node $m$ is given according to Shannon's formula as:
 \begin{align}\label{eq:energy_con}
 E_{m}(t) = \frac{\sigma^2}{g_{u,m}(t)}\left(2^{S/B} - 1\right).
 \end{align}
 
 Clearly, when node $m$ is scheduled to transmit an update packet at time instant $t$, its current battery level $e_{m}(t)$ should be at least equal to $E_{m}(t)$. Therefore, the energy level at node $m$ is updated as $e_{m}(t) \triangleq e_{m}(t) - E_{m}(t),\, \forall i: t = t_{i,m}$.

 \subsection{Problem Formulation}
 Our goal is to characterize the \emph{age-optimal policy} which determines the UAV's velocity and the node scheduled for transmission at every time instant over a finite horizon of time $\tau$. Let $\boldsymbol{t}_m \triangleq \left[t_{1,m}, \dots, t_{n_m,m}\right]^T$ be an ordered vector that contains the time instants during which node $ m $ transmits its update packets to the UAV. Then, a policy $ \pi $ consists of $ v_{u,x}(t) $ and $ v_{u,y}(t) $, for all $t\in [0,\tau]$, and $ \boldsymbol{t}_m $ for all $ m \in \mathcal{M} $. The objective of the age-optimal policy is to minimize the NWAoI defined as follows:
 \begin{align}
 	G(\boldsymbol{t}_1,\dots, \boldsymbol{t}_M) \triangleq \frac{2}{\tau^2} \sum_{m=1}^M\left(\lambda_{m}\int_{0}^{\tau}A_m(t)dt\right),\label{eq:obj}
 \end{align}
 where $\frac{2}{\tau^2}$ is a normalization factor since for a given value of $n_m, \forall m \in \mathcal{M}$, we will have $0< G(\boldsymbol{t}_1,\dots, \boldsymbol{t}_M) \leq \frac{\tau^2}{2}$. Also, $\lambda_{m}\geq 0 $ is the importance weight of the process observed by node $m$ with $ \sum_{m = 1}^{M} \lambda_{m} = 1 $. Every term of the sum in \eqref{eq:obj} can be simplified as follows:
 \begin{align}
 	\int_{0}^{\tau}A_m(t)dt &= \sum_{i=1}^{n_m}\int_{t_{i-1,m}}^{t_{i,m}}A_m(t) dt + \int_{t_{n_m,m}}^{\tau} A_m(t)dt\nonumber \\
 	&=\sum_{i=1}^{n_m}\left[ A^{\textrm{min}}_{m} (t_{i,m} - t_{i-1,m}) + \frac{(t_{i,m} - t_{i-1,m})^2}{2}\right]
 	+ A^{\textrm{min}}_{m} (\tau - t_{n_m,m})+ \frac{(\tau - t_{n_m,m})^2}{2}\nonumber\\
 	 &=A^{\textrm{min}}_{m} \tau + \sum_{m=1}^{n_m +1}\frac{(t_{i,m} - t_{i-1,m})^2}{2}, \label{eq:integ}
 \end{align}
 such that $ t_{n_m+1,m} \triangleq \tau $. From \eqref{eq:integ}, we can see that $ A^{\textrm{min}}_{m} $ is a fixed value that will have no impact on the optimal solution. Thus, we remove $ A^{\textrm{min}}_{m} $ from \eqref{eq:integ} and define a modified NWAoI as follows:
 \begin{align}\label{eq:objmodified}
 	\bar{G}(\boldsymbol{t}_1,\dots, \boldsymbol{t}_M) \triangleq \frac{1}{\tau^2}\sum_{m=1}^{M}\lambda_{m} \sum_{i = 1}^{n_m +1}(t_{i,m} - t_{i-1,m})^2.
 \end{align}

 Hence, our goal is to find a policy that minimizes the NWAoI in \eqref{eq:objmodified} considering the time, location, speed, and energy constraints, which translates into the following optimization problem:
 \begin{align}
 	&\min_{v_{u,x}(t),v_{u,y}(t), \boldsymbol{t}_1,\dots, \boldsymbol{t}_M} \bar{G}(\boldsymbol{t}_1,\dots \boldsymbol{t}_M),\label{eq:problem}\\
 	\textrm{s.t.}&  \sum_{i = 1}^{n_m}E_m(t_{i,m})\leq E^{\textrm{max}}_{m},\,\, \forall m \in \mathcal{M},\label{eq:energyConst}\\
 	& L_u(0) = L^i_u,\label{eq:init}\\
 	& L_u(\tau) = L^f_u, \label{eq:final}\\
 	&\frac{{\rm d}x_u(t)}{{\rm d}t} = v_{u,x}(t),\,\, \frac{{\rm d}y_u(t)}{{\rm d}t} = v_{u,y}(t),\label{eq:loc_speed_const}\\
	&-v^{\textrm{max}}_{x}\leq v_{u,x}(t)\leq v^{\textrm{max}}_{x},\,\, -v^{\textrm{max}}_{y}\leq v_{u,y}(t)\leq v^{\textrm{max}}_{y}.\label{eq:max_speed_const}
\end{align}

Constraint \eqref{eq:energyConst} comes from the fact that each node's total energy consumption for packet transmissions is constrained by its total available energy. The constraints on the initial and final location of the UAV are represented by \eqref{eq:init} and \eqref{eq:final} whereas the UAV's velocity constraints are represented by \eqref{eq:loc_speed_const} and \eqref{eq:max_speed_const}. Solving \eqref{eq:problem} is challenging because the number of times each node transmits its update packets is an unknown variable, thus, \eqref{eq:problem} needs to be solved for each choice of $n_m$ to obtain the minimum NWAoI. In addition, the constraints on the UAV's speed as well as its initial and final locations must be satisfied by the UAV's trajectory, and an energy constraint is required to be satisfied for each node. Therefore, \eqref{eq:problem} is a constrained mixed-integer problem which is challenging to solve \cite{optimization_book}. To this end, we provide a relaxation on the problem that helps us to derive the exact optimal solution using a convex optimization-based approach.

\section{Convex Optimization-based Age-optimal Trajectory}
In order to relax the problem in \eqref{eq:problem}, let us consider fixed values for $ n_{1}, \dots,n_{M} $. In other words, we will now solve problem \eqref{eq:problem} assuming that we know how many times each node should send their update packets to the UAV. In Section \ref{sec:NCRL}, we will provide an algorithm to find the optimal values for $ n_{1}, \dots,n_{M} $. We define a mapping $ \Upsilon: \boldsymbol{t}_1,\dots,\boldsymbol{t}_M \mapsto \boldsymbol{t}_u$ which maps the time instants for packet updates of each node to a sequence $ \boldsymbol{t}_u = \{t_1,\dots,t_{n}\} $ such that $ n \triangleq \sum_{i=1}^M n_i $ and $ t_i \leq t_{i+1},\,\forall i = 1,\dots,n $. Mapping $ \Upsilon $ indicates the order with which the nodes must transmit their packets to the UAV. For instance if $ t_{i,j} $ is mapped to $ t_k $ and $ t_{l,q} $ is mapped to $ t_{k+1} $, then node $ j $ transmits its $ i $-th update packet to the UAV before node $ q $ transmits its $ l $-th update packet to the UAV. 

We define $ \boldsymbol{x}_u \triangleq \left[x_{u,1},\dots,x_{u,n}\right]^T $ and $ \boldsymbol{y}_u \triangleq \left[y_{u,1},\dots,y_{u,n}\right]^T $ such that $ x_{u,i} \triangleq  x_u(t_i)$ and  $  y_{u,i} \triangleq y_u(t_i) $, $ 1 \leq i \leq n $. Here, $ (x_{u,0} , y_{u,0}) $ represents the initial location of the UAV, and $ (x_{u,n+1} , y_{u,n+1}) $ represents the final location of the UAV. We define $ t_0 = 0 $ and $ t_{n+1} = \tau$. Now, from \eqref{eq:loc_speed} we can write:
\begin{align}
	x_{u,i+1} - x_{u,i} = \int_{t_i}^{t_{i+1}} v_{u,x}(t) dt, \forall i \in \left\{0,\dots, n\right\}, \label{eq:Xconstraint} \\
	y_{u,i+1} - y_{u,i} = \int_{t_i}^{t_{i+1}} v_{u,y}(t) dt, \forall i \in \left\{0,\dots, n\right\},\label{eq:Yconstraint}
\end{align}
such that \eqref{eq:Xconstraint} and \eqref{eq:Yconstraint} are feasible if: 
\begin{align}
	|x_{u,i+1} - x_{u,i}|&\leq v^{\textrm{max}}_{x} (t_{i+1} - t_i),\label{eq:xspeedlimit}\\ |y_{u,i+1} - y_{u,i}|&\leq v^{\textrm{max}}_{y} (t_{i+1} - t_i).\label{eq:yspeedlimit}
\end{align}
 Equations \eqref{eq:xspeedlimit} and \eqref{eq:yspeedlimit} indicate that the distance between the UAVs' location in two consecutive time instants is constrained due to the UAVs' speed limitations in \eqref{eq:max_speed_const}. For example, if \eqref{eq:xspeedlimit} and \eqref{eq:yspeedlimit} are satisfied, one solution can be $ v_{u,x}(t) = \frac{x_{u,i+1}-x_{u,i}}{t_{i+1}-t_i} $ and $ v_{u,y}(t) = \frac{y_{u,i+1}-y_{u,i}}{t_{i+1}-t_i} $.
 
In addition, let $ \boldsymbol{x}_m \triangleq \left[x_{m,1} ,\dots,x_{m,n_m}\right]^T $ and $ \boldsymbol{y}_m \triangleq \left[y_{m,1},\dots,y_{m,n_m}\right]^T $ such that $ x_{m,i} \triangleq x_{u}(t_{i,m}) $ and $ y_{m,i} \triangleq y_{u}(t_{i,m}) $. Note that in this case, the mapping $ \Upsilon $ maps $ \boldsymbol{x}_1, \dots, \boldsymbol{x}_M $ to $ \boldsymbol{x}_u $, and maps $ \boldsymbol{y}_1, \dots, \boldsymbol{y}_M $ to $ \boldsymbol{y}_u $. Now, we can express node $ m $'s energy requirement for constraint \eqref{eq:energyConst} as $
 	\sum_{t \in \mathcal{T}_m}E_m(t) = \frac{\sigma ^2 \left(2^{S/B} - 1\right)}{\beta_0}\sum_{i = 1 }^{n_m}\left[\left(x_{m,i}-x_m\right)^2 + \left(y_{m,i}-y_m\right)^2 + h^2\right] $. Moreover, we define $ c_m \triangleq \frac{E^{\textrm{max}}_{m}\beta_0}{\sigma ^2 \left(2^{S/B} - 1\right)} - n_mh^2 $. Next, we can express the problem in \eqref{eq:problem} for a given scheduling policy (order of updates) as follows:
 	\begin{align}
 		&\min_{\boldsymbol{x}_u,\boldsymbol{y}_u,\boldsymbol{t}_u}\bar{G}(\boldsymbol{t}_u),\label{eq:objrelaxed}\\
 		\textrm{s.t.}&\sum_{i = 1}^{n_m}\left[\left(x_{m,i}-x_m\right)^2 + \left(y_{m,i}-y_m\right)^2 \right] \leq c_m,\, \forall m\in \mathcal{M},\label{eq:energy}\\
 		&|x_{u,i+1} - x_{u,i}|\leq v^{\textrm{max}}_{x} (t_{i+1} - t_i),\, \forall i \in \{0,\dots,n\},\label{eq:xspeedconst}\\ &|y_{u,i+1} - y_{u,i}|\leq v^{\textrm{max}}_{y} (t_{i+1} - t_i),\, \forall i \in \{0,\dots,n\},\label{eq:yspeedconst}\\
 		\label{eq:timeconst}
 		&0\leq t_i\leq\tau,\, \forall i \in \{1,\dots,n\}.
 	\end{align}
\begin{lemma}
The problem in \eqref{eq:objrelaxed} is a convex optimization problem.
\end{lemma}	
\begin{proof}
	The term $ \sum_{i=1}^{n_m +1}(t_{i,m} - t_{i-1,m})^2 $ in \eqref{eq:objrelaxed} can be expressed as:
	\begin{align}
		\sum_{1}^{n_m +1}(t_{i,m} - t_{i-1,m})^2  = \boldsymbol{t}_m^T \boldsymbol{Q}_m \boldsymbol{t}_m + \tau^2 - 2\tau t_{n_m},
	\end{align}
	such that:
	\begin{align}
		\boldsymbol{Q}_m = \left[
			\begin{array}{c c c c c}
				2 & -1 & 0 & \cdots & 0\\
				-1 & 2 & -1 & \ddots & \vdots\\
				0 & -1 & \ddots & \ddots &0\\
				\vdots & \ddots & \ddots & \ddots & -1 \\
				0 & \cdots & 0 & -1 & 2
			\end{array}
		\right].
	\end{align}
	$ \boldsymbol{Q}_m $ is a diagonally dominant matrix meaning that the magnitude of the diagonal entry in a row is larger than or equal to the sum of the magnitudes of all the other (non-diagonal) entries in that row. Moreover, $ \boldsymbol{Q}_m $ is symmetric and its diagonal entries are positive. Therefore, $ \boldsymbol{Q}_m $ is a positive definite matrix. Hence, for any $ m \in \mathcal{M} $, $  \sum_{1}^{n_m +1}(t_{i,m} - t_{i-1,m})^2 $ is convex thus, \eqref{eq:objrelaxed} is convex \cite{optimization_book}. The left hand side of the condition in \eqref{eq:energyConst} can be written as $ \sum_{i = 1}^{n_m}\left[\left(x_{m,i}-x_m\right)^2 + \left(y_{m,i}-y_m\right)^2\right] = \boldsymbol{x}_m^T \boldsymbol{I}_{n_m}\boldsymbol{x}_m + \boldsymbol{y}_m^T\boldsymbol{I}_{n_m}\boldsymbol{y}_m - 2x_m\sum_{i=1}^{n_m}x_{i,m} - 2y_m\sum_{i=1}^{n_m}y_{i,m} + n_mx_m^{n_m} + n_my_m^{n_m}$ where $ \boldsymbol{I}_{n_m} $ is an identity matrix with $ n_m $ columns and rows. Since $ \boldsymbol{I}_{n_m} $ is an identity matrix, then, $ \boldsymbol{x}_m^T \boldsymbol{I}_{n_m}\boldsymbol{x}_m $ and $ \boldsymbol{y}_m^T\boldsymbol{I}_{n_m}\boldsymbol{y}_m $ are convex terms. We now see that, $ - 2x_m\sum_{i=1}^{n_m}x_{i,m} $ and $ - 2y_m\sum_{i=1}^{n_m}y_{i,m} $ are linear terms and $ n_mx_m^{n_m} + n_my_m^{n_m} $ is a constant value. Therefore, the constraint in \eqref{eq:energyConst} is convex. Meanwhile, the constraints in \eqref{eq:xspeedconst}-\eqref{eq:timeconst} are linear. Therefore, \eqref{eq:objrelaxed} is a convex optimization problem which completes the proof.
 \end{proof}
 
Moreover, for some special cases, we can derive a closed-form expression for the minimum NWAoI, as shown next. 
\subsection{NWAoI Lower Bound Analysis}
A lower bound on the minimum NWAoI can be derived by considering no limits on the UAV's speed. To derive this lower bound value, we define $ \bar{n}_m \triangleq \left\lfloor\frac{E^{\textrm{max}}_{m}\beta_0}{\sigma ^2 \left(2^{S/B} - 1\right)h^2}\right\rfloor $ as the maximum number of times that node $ m $ can send update packets since if the UAV stays on top of node $ m $, it requires exactly $ \frac{\sigma ^2 \left(2^{S/B} - 1\right)h^2}{\beta_0} $ amount of energy for each update transmission.
\begin{theorem}\label{theorem:lowerbound}
 A lower bound on the minimum NWAoI can be expressed as follows:
	\begin{align}\label{eq:minNWAoI}
		\bar{G} \geq \bar{G}_{\rm min} \triangleq\sum_{m=1}^{M}\frac{\lambda_{m}}{\bar{n}_m+1}.
	\end{align}
\end{theorem}
\begin{proof}
	See Appendix \ref{App:Theorem1}.
\end{proof}
\begin{remark} 
Theorem \ref{theorem:lowerbound} shows that the optimal scheduling policy that results in the lower bound on NWAoI in \eqref{eq:minNWAoI} is the one that updates every node $m$ periodically after every $ \frac{\tau}{\bar{n}_m + 1} $  seconds. Moreover, we can see from \eqref{eq:minNWAoI} that, since $ \bar{n}_m $ is linearly dependent on $ E^{\textrm{max}}_{m} $, the nodes with lower battery capacities can have a higher impact on the NWAoI. This can be helpful in deciding on the node importance values, $ \lambda_m $. For instance, in order to have an equal impact from each node, $ \lambda_m $ can be chosen to be proportional to $ \bar{n}_m +1 $.
\end{remark}
Although Theorem \ref{theorem:lowerbound} provides a lower bound on the minimum NWAoI, this lower bound value may not be achievable in practice because we did not account for the speed limitations of the UAV while deriving this bound. That said, it is natural to wonder about the minimum speed of the UAV required to achieve the bound in \eqref{eq:minNWAoI}, which is studied next. The main idea is that the UAV receives the updates from the nodes not exactly on top of them but at a small distance away from them (by using the residual of the energy left from the floor operation in finding $\bar{n}_m$), which reduces the distance between two update locations and, hence, minimizes the required speed. In particular, the minimum speed requirement that allows the UAV to achieve the lower bound in \eqref{eq:minNWAoI} is the solution of the following optimization problem:
\begin{align}\label{eq:minSpeed}
	&v_{\textrm{min}} = 	\min_{\boldsymbol{x}_u,\boldsymbol{y}_u} v\\
	\textrm{s.t.}&\,\,t_{i,m} = \frac{i\tau}{\bar{n}_m+1},\, \forall m \in \{ 1,\dots,M\}, \, i \in \{1,\dots,\bar{n}_m\},\\\label{eq:energyConst2}
	&\sum_{i = 1}^{\bar{n}_m}\left[\left(x_{m,i}-x_m\right)^2 + \left(y_{m,i}-y_m\right)^2 \right] \leq c_m,\, \forall m\in \mathcal{M},\\\label{eq:xspeedconst2}
	&|x_{i+1} - x_{i}|\leq v (t_{i+1} - t_i),\, \forall i \in \{ 0,\dots,\sum_{m=1}^{M}\bar{n}_m \},\\ &|y_{i+1} - y_{i}|\leq v (t_{i+1} - t_i),\, \forall i \in \{0,\dots,\sum_{m=1}^{M}\bar{n}_m \}.\label{eq:yspeedconst2}
\end{align}

In problem \eqref{eq:minSpeed}, we consider that the update time instants are known and set to be the ones derived in Theorem \ref{theorem:lowerbound}. The solution should satisfy the node's energy and UAV's location constraints in \eqref{eq:energyConst2}-\eqref{eq:yspeedconst2}. Also, in \eqref{eq:xspeedconst2} and \eqref{eq:yspeedconst2}, we consider that the maximum allowable speed of the UAV in directions $ x $ and $ y $ are equal which is a practical assumption because the UAV's motors are usually identical. It can be easily shown that the problem in \eqref{eq:minSpeed} is a linear program with convex constraints which can be solved using interior point techniques \cite{optimization_book}. However, the solution may not give us a closed-form expression on the minimum required speed for the UAV. A closed-form expression could be helpful in choosing the type of UAV or defining the parameters of the optimization problem, especially the node weights. Therefore, in the following, we derive a closed form expression for the upper bound on the UAV's minimum required speed. To this end, we define a \emph{scheduling policy}, $ \boldsymbol{u} $, which is a vector that contains the indices of the scheduled nodes and is ordered based on the scheduled time instants of the nodes. For instance, letting $ u_i $ and $ u_{i+1} $ be the $ i $-th and $ i+1 $-th elements of $ \boldsymbol{u} $, the node $ u_i $ will be scheduled for transmission one step prior to $ u_{i+1} $. Note that for every $ \boldsymbol{u}, $ there exists a vector $ \boldsymbol{t}_u $, and, hence, we will have $ \bar{G}(\boldsymbol{t}_u) \equiv \bar{G}(\boldsymbol{u})  $. Also let us define $ \bar{\boldsymbol{u}} $ as the scheduling policy that keeps the order of updates for the optimal time instants derived in Theorem \ref{theorem:lowerbound}. In the following, we derive the upper bound for the UAV's minimum required speed.
\begin{prop}\label{proposition:vmin}
	If no two nodes $ m $ and $ p $ exist such that $ \bar{n}_m + 1 $ is a divisor of $ \bar{n}_p + 1 $ or vice versa, then the UAV's minimum speed needed to achieve the minimum NWAoI is upper bounded by:
	\begin{align}\label{eq:vmin}
		{v}_{\textrm{min}} \leq \bar{v}_{\textrm{min}} \triangleq \max_{i} \left\{\frac{\left|y_{\bar{u}_i+1} - y_{\bar{u}_i}\right|}{t_{i+1} - t_i},\frac{\left|x_{\bar{u}_i+1} - x_{\bar{u}_i}\right|}{t_{i+1} - t_i},\,\, \forall\, i = 0,\dots,\sum_{m=1}^{M}\bar{n}_m\right\},
	\end{align}
	such that $\forall i \in \{0,\dots,\sum_{m=1}^{M}\bar{n}_m \}$, $ \{t_i\} $ are the time instants derived in Theorem \ref{theorem:lowerbound}.
\end{prop}
\begin{proof}
	Let $ x_{m,i} = x_m $ and $ y_{m,i} = y_m $ for $ m \in \mathcal{M}$ and $i \in \{1,\dots,\bar{n}_m \}$, meaning that the UAV updates the nodes when it is on top of them. Then, the UAV needs to travel between the top of two nodes in less than the difference between two consecutive time instants. Therefore, the distances covered by the UAV between two consecutive updates in $ x $ and $ y $ directions will be $ \left|x_{\bar{u}_i+1} - x_{\bar{u}_i}\right| $ and $ \left|y_{\bar{u}_i+1} - y_{\bar{u}_i}\right| $, respectively. Moreover, since $ t_i $ is the time instant of the $ i $-th update using the policy $ \bar{\boldsymbol{u}} $, then the speed requirements for the travel before $ i $-th update are $ \frac{\left|x_{\bar{u}_i+1} - x_{\bar{u}_i}\right|}{t_{i+1}-t_i} $ and $ \frac{\left|y_{\bar{u}_i+1} - y_{\bar{u}_i}\right|}{t_{i+1}-t_i} $. Therefore, the UAV's speed has to be at least the maximum value of the required speed for all travels, which yields \eqref{eq:vmin}. However, if there exists a time instant $ t_i $ such that $ t_i = t_{i+1} $, then UAV's speed tends to be infinity which is infeasible. Therefore, we need to have $ \frac{i\tau}{\bar{n}_m+1} \neq \frac{j\tau}{\bar{n}_p+1} $, for all pairs of nodes $ m $ and $ p $. To this end, no two nodes $ m $ and $ p $ must exist such that $ \bar{n}_m + 1 $ is divisor of $ \bar{n}_p + 1 $ or vice versa, which completes the proof.
\end{proof}
Proposition \ref{proposition:vmin} derives a minimum value on UAV's speed so as to guarantee achieving the lower bound on NWAoI if any two nodes do not have equal update time instants. If the UAV needs to update two nodes at exactly the same time instant, then the required speed can be derived by solving the problem in \eqref{eq:vmin}. If problem \eqref{eq:vmin} does not yield a solution then the lower bound NWAoI is not achievable. In this case, a policy different than $ \bar{\boldsymbol{u}} $ should be fed to the problem in \eqref{eq:objrelaxed} to find the update time instants and locations. 

Although problem \eqref{eq:objrelaxed} can be solved, it requires the knowledge of scheduling policy, i.e., each node's number and order of updates. However, finding the scheduling policy is challenging especially when the nodes are equipped with batteries of large capacities since the nodes may send updates more frequently in such case. In fact, for known values of $ n_{1}, \dots,n_{M} $, there exists $ \frac{\left(\sum_{m=1}^{M} n_m\right) !}{\prod_{m=1}^{M}n_m!} $ different orders for updating nodes. Therefore, using a brute force method, the number of times one should solve \eqref{eq:objrelaxed} to find the optimal solution for the original problem in \eqref{eq:problem} is: 
\begin{align}\label{eq:maxupdate}
	\sum_{n_1 = 1}^{\bar{n}_1} \cdots \sum_{n_M = 1}^{\bar{n}_M} \frac{\left(\sum_{m=1}^{M} n_m\right) !}{\prod_{m=1}^{M}n_m!}.
\end{align}

From \eqref{eq:maxupdate}, we can see that finding the optimal scheduling policy using brute force has a combinatorial form which is computationally expensive. Hence, in the following, we propose a similar NCRL method to that in \cite{bello2016neural} and \cite{khalil2017learning} in order to find the optimal scheduling policy for the nodes without using brute force.

\section{Neural Combinatorial Based Deep Reinforcement Learning for Optimal Scheduling}\label{sec:NCRL}
In order to find the optimal scheduling policy for the nodes, we first propose an NCRL approach \cite{bello2016neural} and \cite{khalil2017learning}. Unlike the DRL solution proposed in our early work \cite{UAV_AoI_DRL1} in which an \emph{environment} is defined as the area within which the nodes are located, our proposed NCRL considers the problem in \eqref{eq:objrelaxed} as an environment that receives a policy $ \boldsymbol{u} $ and outputs the NWAoI, $ \bar{G}(\boldsymbol{u}) $. In particular, we consider three main elements for this problem: \emph{state} of the environment, \emph{action} of the UAV, and the \emph{reward} from the environment as described in the following.
\subsection{State, Action, Reward, and Optimal Scheduling Policy Definition}
The state of the environment can be defined as a matrix $ \boldsymbol{S}_{n} $ that has $ n +1 $ columns: 1) the first column contains the initial battery levels and the time instant of operation, and 2) every column after the first column contains the energy levels of the nodes after an update as well as the time instant of that update. In other words, for an update policy $ \boldsymbol{u} $, the $(i-1)$-th column of $ \boldsymbol{S}_n $ represents the energy levels of the nodes before node $ u_i $ is updated. Formally, the $ i $-th column of  $ \boldsymbol{S}_{n} $, will be:
\begin{align}\label{eq:state}
	\boldsymbol{s}_{n,i} \triangleq \left[E_1(t_i),\dots,E_M(t_i),t_i\right]^T.
\end{align}

Furthermore, the initial state is defined as $ \boldsymbol{S}_{0} = \boldsymbol{s}_{0,1} \triangleq \left[E^{\textrm{max}}_{1},\dots,E^{\textrm{max}}_{M},0\right]^T $ which captures the available energy of the nodes in the beginning of the problem where the first time instant is set to be $ 0 $. Also, note that $ E_m(t_i) $ and $ t_i $ can be obtained for $ 0 \leq i \leq n $ and $ 1 \leq m \leq M $ by solving the problem in \eqref{eq:objrelaxed} using the scheduling policy vector $ \boldsymbol{u} $. Therefore, the state space of this problem is the space of all 2-D matrices with $m+1$ rows such that any element at row $m $ for $m\in \mathcal{M}$ is in $\left[0,E^{\textrm{max}}_{m}\right]$ and any element at row $M+1$ is in $\left[0,\tau\right]$.

At any state of the problem, the UAV can either choose to schedule a node for sending an update packet or terminate the policy. Therefore, an action $ a_n $ at state $ \boldsymbol{S}_{n-1} $ can get any integer value in the action set $ \mathcal{A} \triangleq \left\{0,\dots,M\right\} $, such that $ a_n = m > 0 $ means that the node $ m $ is scheduled for transmission; $ a_n = 0 $ terminates the policy, i.e., no new update transmissions will be added to the current policy. Let $ \boldsymbol{u}_{n-1} $ be a policy that contains $ n-1 $ node indexes such that it transitions state $ \boldsymbol{S}_0 $ to $ \boldsymbol{S}_{n-1} $. Then, at every state $ \boldsymbol{S}_{n-1} $, action $ a_n > 0 $ transitions $ \boldsymbol{S}_{n-1} $ to $ \boldsymbol{S}_{n} $ such that $ \boldsymbol{S}_{n} $ is the transition from $ \boldsymbol{S}_{0} $ using policy $ \boldsymbol{u}_{n} \triangleq [\boldsymbol{u}_{n-1} , a_n]$. In other words, at every state of the problem, the UAV adds a node to the end of the scheduling policy, solves the problem in \eqref{eq:objrelaxed}, and transitions the state of the problem to a new one. While transitioning the state of the problem, the UAV receives a new NWAoI value from \eqref{eq:objrelaxed} and uses it as a \emph{reward} to derive the optimal scheduling policy. In particular, we define the reward for every action as the reduction in the NWAoI value, which can be expressed as:
\begin{align}\label{eq:reward_def}
	r_n(\boldsymbol{S}_{n-1},\boldsymbol{u}_{n-1},a_n) = \bar{G}(\boldsymbol{u}_{n-1}) - \bar{G}(\boldsymbol{u}_n).
\end{align}

We also define $ \bar{G}(\boldsymbol{u}_{0}) = 1 $ since when policy $ \boldsymbol{u} $ is empty, i.e., none of the nodes will be scheduled for update transmissions in that case, and, hence, the NWAoI will have a maximum value of 1. Furthermore, we consider that the reward of the termination action $ a _n= 0 $ is 0, i.e., $r_n(\boldsymbol{S}_{n-1},\boldsymbol{u}_{n-1},a_n = 0) = 0 $. Using the definition of the reward in \eqref{eq:reward_def}, we can see that the NWAoI for a policy $ \boldsymbol{u}_n $ can be expressed as:
\begin{align}\label{obj_DRL}
	\bar{G}(\boldsymbol{u}_n) = 1 - \sum_{k = 1}^{n} r_k(\boldsymbol{S}_{k-1},\boldsymbol{u}_{k-1},a_k).
\end{align}

Therefore, the optimal policy that minimizes \eqref{obj_DRL} (which is also the objective function of the problem in \eqref{eq:objrelaxed}) can be written as follows:
\begin{align}\label{eq:Optim_policy}
	\boldsymbol{u}^\star = \argmax_{n,u_n}\sum_{k = 1}^{n} r_k(\boldsymbol{S}_{k-1},\boldsymbol{u}_{k-1},a_k).
\end{align}

 Owing to the nature of evolution of the problem, represented by $ \boldsymbol{S}_{n-1} $, $ a_n $, $ \boldsymbol{u}_n $, $ \boldsymbol{S}_{n} $, and $ r_n(\boldsymbol{S}_{n-1},\boldsymbol{u}_{n-1},a_n) $, the problem can be modeled as a finite-horizon MDP with finite state and action spaces. However, due to the curse of extremely high dimensionality in the state space, it is computationally infeasible to obtain $\boldsymbol{u}^\star$ using the standard finite-horizon DP algorithm \cite{powell2007approximate}. Motivated by this, we propose next a deep RL algorithm for solving (\ref{eq:Optim_policy}). Deep RL is suitable here because it can reduce the dimensionality of the large state space while learning the optimal policy at the same time using neural combinatorial optimization methods as in \cite{bello2016neural} and \cite{khalil2017learning}.

\subsection{Deep Reinforcement Learning Algorithm}\label{sec:DQN}
The proposed deep RL algorithm has two components:
(i) an artificial neural network (ANN), that reduces the dimension of the state space by extracting its useful features and (ii) an RL component, which is used to find the best policy based on the ANN's extracted features, as shown in Fig. \ref{fig:deepRL}. To derive the policy that maximizes the total expected reward of the system, we use a $Q$-learning algorithm \cite{powell2007approximate}. In this algorithm, we define a state-action value function $Q(\boldsymbol{S}_{n-1},a_n)$ which is the expected reward of the system starting at state $\boldsymbol{S}_{n-1}$, performing action $a_n$ and following policy $\boldsymbol{u}$. In $Q$-learning algorithm, we try to estimate the $Q$-function using any policy that maximizes the future reward. To this end, we use the so-called Bellman update rule:
\begin{align}\label{eq:Bellman}
Q_{k+1}\left(\boldsymbol{S}_{n-1},a_n\right) = Q_{k}\left(\boldsymbol{S}_{n-1},a_n\right) + \beta \Big(r_n(\boldsymbol{S}_{n-1},\boldsymbol{u}_{n-1},a_n)+ \gamma \max_{\alpha} Q_{k}\left(\boldsymbol{S}_{n},\alpha\right) - Q_{k}\left(\boldsymbol{S}_{n-1},a_n\right) \Big),
\end{align}
where $\beta$ is the learning rate, and $\gamma$ is a discount factor. The discount factor can be set to a value between 0 and 1 if the UAV's task is \emph{continuing} which means the task will never end, and, hence, the current reward will have a higher value compared to the unknown future reward. However, we have here two terminal cases: 1) when problem \eqref{eq:objrelaxed} does not have a solution for a scheduling policy $ \boldsymbol{u}_n $ and 2) when $ a_n = 0 $ ( the policy is terminated). Therefore, our problem is \emph{episodic}, and so we set $\gamma = 1$. This aligns with the optimal policy definition in \eqref{eq:Optim_policy} in which all of the steps of an episode until the terminal state have equal weights in the evaluation of the policy.
\begin{figure}[t!]
    \centering
    \includegraphics[width=0.55\columnwidth]{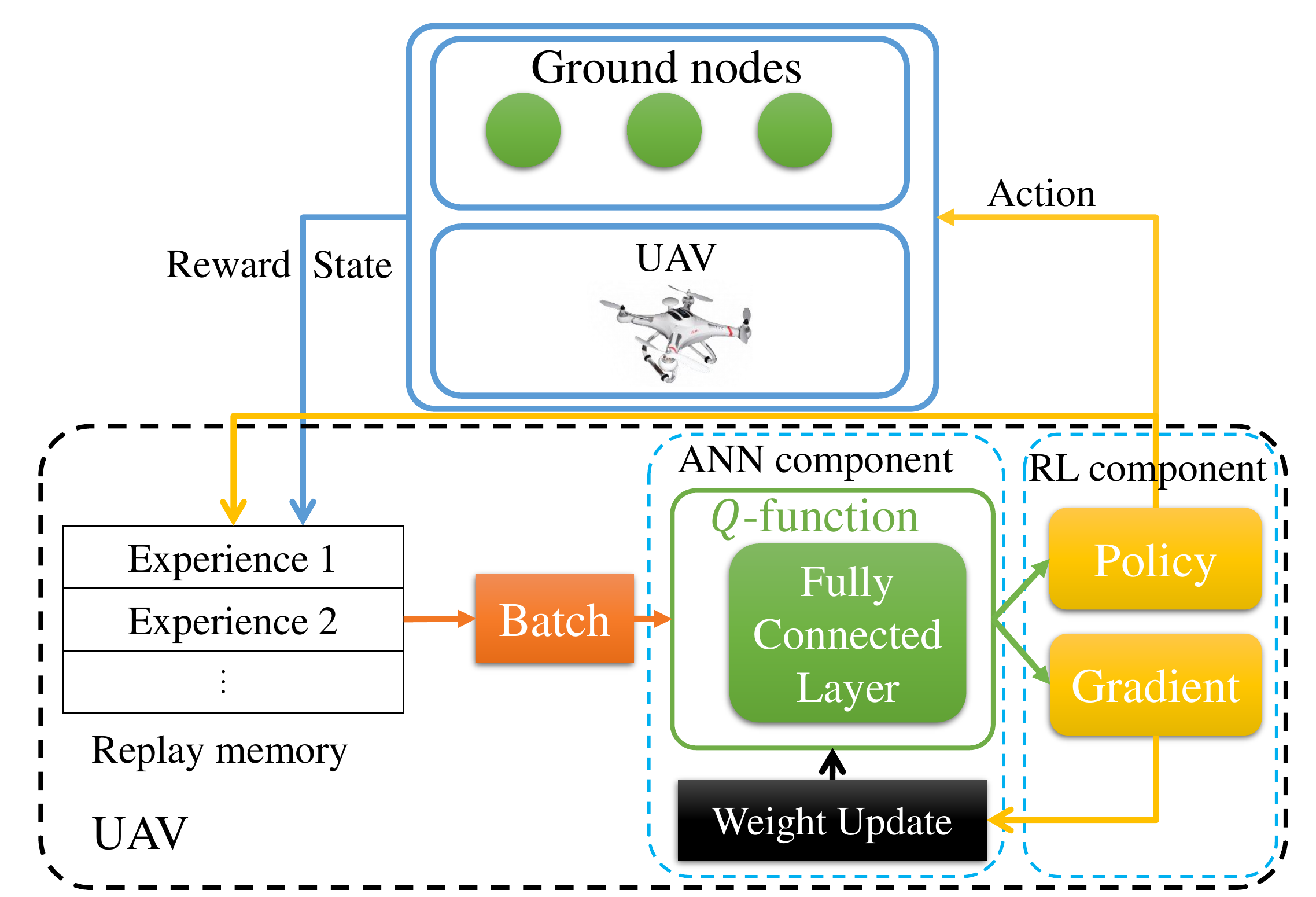}
    \caption{The deep RL architecture.}
    \label{fig:deepRL}
\end{figure}

Since, using \eqref{eq:Bellman}, the UAV always has an estimate of the $Q$-function, it can \emph{exploit} the learning by taking the action that maximizes the reward. However, when learning starts, the UAV does not have confidence on the estimated value of the $Q$-function since it may not have visited some of the state-action pairs. Thus, the UAV has to \emph{explore} the environment (all state-action pairs) to some degree. To this end, an $\epsilon$-greedy approach is used where $\epsilon$ is the probability of exploring the environment at the current state \cite{mnih2015human}, i.e., taking a random action with some probability. Since the need for exploration goes down with time, one can reduce the value of $\epsilon$ to $0$ as the learning goes on to ensure that the UAV chooses the optimal action rather than explore the environment.

The iterative method in \eqref{eq:Bellman} can be applied efficiently for the case in which the state space is small. However, the extremely high dimension of the state space in our problem makes such an iterative approach impractical, since it requires a large memory and will have a slow convergence rate. Also, this approach cannot be generalized to unobserved states, since the UAV must visit every state and take every action to update every state-action pair \cite{powell2007approximate}. Thus, we employ ANNs which are very effective at extracting features from data points and summarizing them to smaller dimensions. We use a DQN approach in \cite{mnih2015human,FerdowsiIoT,FerdowsiLSTMDRL} where the learning steps are the same as in $Q$-learning, however, the $Q$-function is approximated using an ANN $Q(\bar{\boldsymbol{s}},a|\boldsymbol{\theta})$, where $ \bar{\boldsymbol{s}} $ is a vector representation of the state and $\boldsymbol{\theta}$ is the vector containing the weights of the ANN. 

In our problem, the states have a matrix form with fixed number of rows and varying number of columns. However, in order to apply the DQN approach, the state matrix in our problem must be mapped into a vector representation with fixed number of elements. To do so, we propose two methods as follows. First, for scenarios with small number of nodes, in which the size of state matrix $ \boldsymbol{S}_n $ is not very large, the last column of the state $ \boldsymbol{s}_{n,n} $ can be used as the state representation since it captures the final energy levels of the nodes after all of the updates. Second, for large-scale scenarios, there will be spatio-temporal interdependencies between the nodes and their update time instants. Thus, an ANN-based \emph{autoencoder} can be used to map the varying size states to a fixed size vector (which will then be used in the DQN) \cite{sutskever2014sequence}. This autoencoder will be studied in detail in the following section. After deriving the state representation vector, $ \bar{\boldsymbol{s}} $, a fully connected (FC) layer, as in \cite{mnih2015human}, is used to extract abstraction of the state representation. In the FC, every artificial node of a layer is connected to every artificial node of the next layer via the weight vector $\boldsymbol{\theta}$. The goal is to find the optimal values for $\boldsymbol{\theta}$ such that the ANN will be as close as possible to the optimal $Q$-function. To this end, we define a loss function for any set of $\left(\bar{\boldsymbol{s}}_{n},a_n,r_n,\bar{\boldsymbol{s}}_{n+1}\right)$, as follows:
\begin{align}
    L(\boldsymbol{\theta}_{k+1}) = \Big[r_n + \gamma \max_{\alpha'} Q(\bar{\boldsymbol{s}}_{n},\alpha'|\boldsymbol{\theta}_{k}) - Q(\bar{\boldsymbol{s}}_{n-1},a_n|\boldsymbol{\theta}_{k+1})\Big]^2,
\end{align}
where subscript $k+1$ is the episode at which the weights are updated.
In addition, we use a \emph{replay memory} that saves the evaluation of the state, action, and reward of past \emph{experiences}, i.e., past state-actions pairs and their resulting rewards. Then, after every episode, we sample a batch of $b$ past experiences from the replay memory and we find the gradient of the weights using this batch as follows:
\begin{align}\label{eq:gradient}
\nabla_{\boldsymbol{\theta}_{k+1}} L(\boldsymbol{\theta}_{k+1}) = \Big[r_n + \gamma \max_{\alpha'} Q(\bar{\boldsymbol{s}}_{n},\alpha'|\boldsymbol{\theta}_{k})- Q(\bar{\boldsymbol{s}}_{n-1},a_n|\boldsymbol{\theta}_{k+1}) \Big]\times\nabla_{\boldsymbol{\theta}_{k+1}}Q(\bar{\boldsymbol{s}}_{n-1},a_n|\boldsymbol{\theta}_{k+1}).
\end{align}

Using this loss function, we train the weights of the ANN, $ \boldsymbol{\theta} $. It has been shown that using the batch method and replay memory improves the convergence of deep RL \cite{mnih2015human}.
Algorithm \ref{Algorithm:DeepRL} summarizes our proposed solution and Fig. \ref{fig:deepRL} shows the architecture of the deep RL algorithm. 
\begin{algorithm}[t]
    \footnotesize
	\caption{Deep RL for NWAoI minimization}
	\begin{algorithmic}[1] 
		\State Initialize a \emph{replay memory} that stores the past experiences of the UAV and an ANN for $ Q $-function. Set $k=1$.
		\State \textbf{Repeat:}
		\State \quad Set $n=1$, initialize an empty policy $ \boldsymbol{u}_0 = []$ and NWAoI $ \bar{G}(\boldsymbol{u}_0) = 1 $ and observe the initial state representation $ \bar{\boldsymbol{s}}_1 $.
		\State \quad \textbf{Repeat:}
		\State \quad \quad Select an action $ a $:
		\State \quad \quad \quad select a random action $a \in \mathcal{A}$ with probability $ \varepsilon $ ,
		\State \quad \quad \quad otherwise select $ a = \argmin_{\alpha} Q(\bar{\boldsymbol{s}}_n,\alpha|\boldsymbol{\theta}_k) $.
		\State \quad \quad Append action $ a $ to the end of policy $ \boldsymbol{u}_{n-1} $ as $ \boldsymbol{u}_{n} = \left[\boldsymbol{u}_{n-1},a\right] $. 
		\State \quad \quad Solve \eqref{eq:objrelaxed} using $ \boldsymbol{u}_{n} $ and find $ \bar{G}(\boldsymbol{u}_n) $.
		\State \quad \quad Observe the reward $ r_n = \bar{G}(\boldsymbol{u}_{n-1}) - \bar{G}(\boldsymbol{u}_n) $ and the new state $ \bar{\boldsymbol{s}}_{n+1} $.
		\State \quad \quad Store \emph{experience} $ \left\{\bar{\boldsymbol{s}}_{n},a_n,r_n,\bar{\boldsymbol{s}}_{n+1}\right\} $ in the replay memory.
		\State \quad \quad $ n = n+1 $
		\State \quad \textbf{Until} $ \bar{\boldsymbol{s}}_{n+1} $ is a terminal state.
		\State \quad Sample a batch of $b$ random experiences $ \left\{\bar{\boldsymbol{s}}_{\eta},a_{\eta},r_{\eta},\bar{\boldsymbol{s}}_{\eta+1}\hspace{-0.5mm}\right\} $ from the replay memory.
		\State \quad Calculate the \emph{target} value $ t $:
		\State \quad \quad If the sampled experience is for terminal state then $ t=r_{\eta} $,
		\State \quad \quad Otherwise $ t=r_{\eta} + \gamma \min_{\alpha'} Q(\bar{\boldsymbol{s}}_{\eta + 1},\alpha'| \boldsymbol{\theta}_k) $.
		\State \quad Derive the gradients for all of the episodes in the batch using \eqref{eq:gradient}.
		\State \quad Train the network $ Q $ using the average of gradients.
		\State \quad $ k = k+1 $.
		\State \textbf{Until} convergence.
	\end{algorithmic}
	\normalsize
	\label{Algorithm:DeepRL}
\end{algorithm}

As already discussed, the proposed DQN approach can work for state representations with fixed number of elements. However, the state of the problem, $ \boldsymbol{S}_n $, in our setup has a matrix form with varying number of columns. Although using the last column of $ \boldsymbol{S}_n $ as the state representation may work in scenarios with small number of nodes, we need to capture spatio-temporal interdependence between the columns of $ \boldsymbol{S}_n $ for large-scale scenarios. Therefore, we next propose a recurrent neural network (RNN) architecture that extracts spatio-temporal interdependencies between the node energy levels and the update time instants in order to feed into the DQN algorithm for such large-scale scenarios.

\subsection{Long Short-Term Memory-based Structure}
We study a special RNN architecture, named LSTM cells \cite{speech2013Graves}, that can learn time interdependence between the columns of the state and map them into a fixed size 1-dimensional state representation. In particular, LSTMs have three main components as shown in Fig. \ref{fig:LSTM}: 1) a forget gate which receives an extra input called the cell state input and learns how much it should memorize or forget from the past, 2) an input gate which aggregates the output of past steps and the current input and passes it through an activation function as done in a conventional RNN, and 3) an output gate which combines the current cell state and the output of input gate and generates the LSTM output \cite{FerdowsiLSTM}. Formally, the relationship between different parts of the LSTM block in Fig. \ref{fig:LSTM} can be expressed as follows:
\begin{align}
	\boldsymbol{f}_i &= \sigma(W_f\left[\boldsymbol{h}_{i-1}^T,\boldsymbol{s}_{n,n-i}^T\right]^T+ \boldsymbol{b}_f),\\
	\boldsymbol{r}_i &= \sigma(W_r\left[\boldsymbol{h}_{i-1}^T,\boldsymbol{s}_{n,n-i}^T\right]^T+ \boldsymbol{b}_r),\\
	\tilde{\boldsymbol{c}}_i &= \tanh\left(W_c\left[\boldsymbol{h}_{i-1}^T,\boldsymbol{s}_{n,n-i}^T\right]^T+ \boldsymbol{b}_c\right)\\
	\boldsymbol{c}_i &= \boldsymbol{f}_i*\boldsymbol{c}_{i-1} + \boldsymbol{r}_i * \tilde{\boldsymbol{c}}_i,\\
	\boldsymbol{o}_i &= \sigma\left(\boldsymbol{W}_o\left[\boldsymbol{h}_{i-1}^T,\boldsymbol{s}_{n,n-i}^T\right]^T+\boldsymbol{b}_o\right)\\
	\boldsymbol{h}_i &= \boldsymbol{o}_i*\tanh(\boldsymbol{c}_i),
\end{align}
where $ \sigma(x) \triangleq \frac{1}{1+e^{-x}} $ is the sigmoid function, $ * $ represents element-wise multiplication, $ \boldsymbol{W}_f $,  $ \boldsymbol{W}_r $, $ \boldsymbol{W}_c $, and $ \boldsymbol{W}_o $ are weight matrices, and $ \boldsymbol{b}_f $, $ \boldsymbol{b}_r $, $ \boldsymbol{b}_c $, and $ \boldsymbol{b}_o $  are bias matrices at the forget, input, and output gates of the LSTM. Given a state $ \boldsymbol{S}_n $, the LSTM uses every column of $ \boldsymbol{S}_n $, $ \boldsymbol{s}_{n,n-i} $ as an input and iteratively calculates an output sequence for $ i \in\{ 1,\dots,n\} $. Next, we show how the cell state and output values can be used as a state representation in our problem.

\begin{figure}[t]
\centering
\includegraphics[width=0.6\columnwidth]{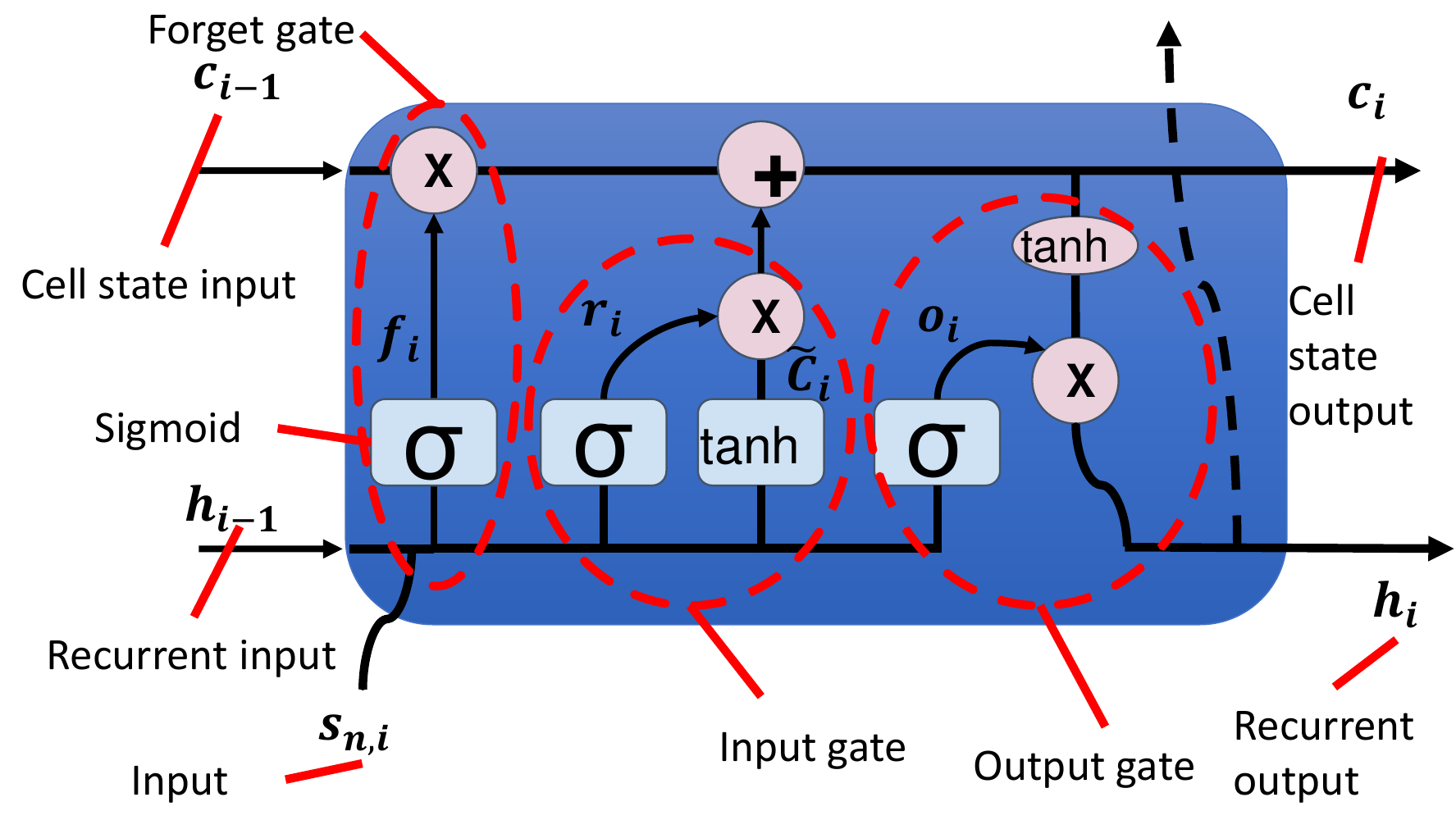}
\caption{A generic LSTM block architecture.}
\label{fig:LSTM}
\end{figure}

\subsection{LSTM-based Autoencoder Using a Sequence-to-Sequence Model}
The LSTM blocks can be used to map the matrix $ \boldsymbol{S}_{n} $ to a vector with fixed size \cite{sutskever2014sequence}. To this end, we use the sequence-to-sequence architecture in Fig. \ref{fig:LSMTAutoEncoder}. Sequence-to-sequence models are commonly used for translation from a language to another language \cite{sutskever2014sequence}. In this architecture, we use two LSTMs: one to receive an input sequence of words (a sentence) in a primary language and one to generate a new sequence of words (a sentence) in a secondary language. Every word in the sequence from primary language is fed to the LSTM iteratively until reaching the last word in the sequence. Then the cell state output, $ \boldsymbol{c}_n $, and recurrent output $ \boldsymbol{h}_n $, are concatenated into a vector, $ \bar{\boldsymbol{s}}_n $. Then, $ \boldsymbol{c}_n $ and $ \boldsymbol{h}_n $ are fed into the second LSTM as the initial cell state and recurrent inputs. Now, the input sequence to the second LSTM will be the sequence of the words from the secondary language. During the training of this model, the goal is to find optimal values for the weights and biases of the LSTMs such that, in essence, $ \bar{\boldsymbol{s}}_n $ represents the \emph{meaning} of the sentence in the primary language. We use the same concept to learn a fixed size representation of our state space as shown in Fig. \ref{fig:LSMTAutoEncoder}. 

In Fig. \ref{fig:LSMTAutoEncoder}, we use $ \boldsymbol{S}_n $ as the input sequence for the first and second LSTMs. In this respect, each column of $ \boldsymbol{S}_n $ represents a word of a sentence in the sequence to sequence model. As a training trick, in \cite{sutskever2014sequence}, the authors show that the last word in the sentence always must be a fixed value that represents the end of sentence. To this end, we train our model by flipping the columns of $ \boldsymbol{S}_n $ left to right. In other words, $ \boldsymbol{s}_{n,n} $ is used as the first input, $ \boldsymbol{s}_{n,n-1} $ is used as the second input and so on. This guarantees the last input to be $ \boldsymbol{s}_{n,1} $ which has fixed values (the energy levels of nodes in the beginning of the problem) as shown in \eqref{eq:state}. We will use the concatenation of vectors $ \boldsymbol{c}_n $ and $ \boldsymbol{h}_n $, as the state representation $ \bar{\boldsymbol{s}}_n $ in our DQN. 

The size of $ \bar{\boldsymbol{s}}_n $ is a hyperparamter of the model which requires to be optimized. To this end, in Algorithm \ref{Algorithm:LSTMAutoencoder}, we propose an iterative method to find the optimal state representation. First we define the weight-based scheduling policy, $ \boldsymbol{u}^\lambda $, as the one that starts with an empty vector, then, keep adding nodes to the policy randomly using a multinomial distribution where the probability of choosing each node $ n $ will be its weight in NWAoI, $ \lambda_{m} $. We use this policy to collect experiences from the problem to train our LSTM autoencoder. In other words, for any $ \boldsymbol{u}^\lambda $, we solve the problem in \eqref{eq:objrelaxed} and derive the state $ \boldsymbol{S}_n $. Afterwards, we use this state to train the model in Fig. \ref{fig:LSMTAutoEncoder}. We train the model using the back propagation method in \cite{speech2013Graves} for different sizes of $ \bar{\boldsymbol{s}}_n $ and choose the size that has the minimum test mean squared error (MSE). Algorithm \ref{Algorithm:LSTMAutoencoder} shows the steps of the training process.

\begin{figure}[t]
	\centering
	\includegraphics[width=\columnwidth]{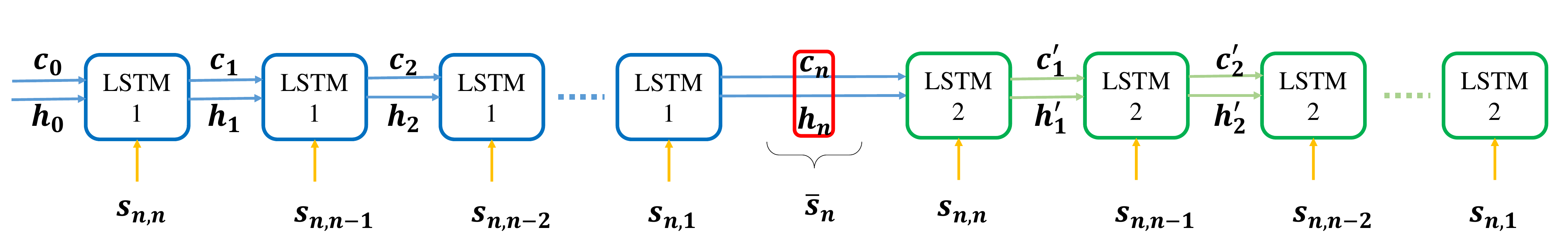}
	\caption{An LSTM-based autoencoder architecture.}
	\label{fig:LSMTAutoEncoder}
\end{figure}

\begin{algorithm}[t]
    \footnotesize
	\caption{Hyperparameter Optimization for LSTM autoencoder-based State Representation}
	\begin{algorithmic}[1] 
		\State Set minimum and maximum sizes, $ k^{\textrm{min}}_{c} $, $ k^{\textrm{max}}_{c} $, $ k^{\textrm{min}}_{h} $, $ k^{\textrm{max}}_{h} $ for the vectors $ \boldsymbol{c}_n $ and $ \boldsymbol{h}_n $ to certain values. Set maximum number of episodes $ \bar{e} $ to a certain value and $e=1$ and initialize a memory that stores the past states of the problem.
		\State Observe the initial state $ \boldsymbol{S}_0 $, and store it to the memory.
		\State \textbf{Repeat:}
		\State \quad Set $n=1$, initialize an empty policy $ \boldsymbol{u}^\lambda_0 = []$, 
		\State \quad \textbf{Repeat:}
		\State \quad \quad Select an action $ a $ using a binomial distribution with probabilities equal to the weights of nodes.
		\State \quad \quad Append action $ a $ to the end of policy $ \boldsymbol{u}^\lambda_{n-1} $ as $ \boldsymbol{u}^\lambda_{n} = \left[\boldsymbol{u}^\lambda_{n-1},a\right] $. 
		\State \quad \quad \quad Solve \eqref{eq:objrelaxed} using $ \boldsymbol{u}_{n} $. If \eqref{eq:objrelaxed} had a solution find $ \boldsymbol{S}_{n+1} $, otherwise, break the loop.
		\State \quad \quad Store state $ \boldsymbol{S}_{n+1} $ in the replay memory.
		\State \quad \quad $ n = n+1 $
		\State \quad e = e + 1.
		\State \textbf{Until} $e+1 = \bar{e}$.
		\State Split the memory randomly into training and test memories with 7 to 3 ratio.
		\State Set $ k_c^\star = k^{\textrm{min}}_{c} $, $ k_h^\star = k^{\textrm{min}}_{h} $, and $ k_c = k^{\textrm{min}}_{c} $. 
		\State \textbf{Repeat}:
		\State \quad Set $ k_h = k^{\textrm{min}}_{h} $.
		\State \quad \textbf{Repeat:}
		\State \quad \quad Initialize an LSTM autoencoder architecture with $ k_c $ and $ k_h $ number of elements for $ \boldsymbol{c}_n $ and $ \boldsymbol{h}_n $.
		\State \quad \quad Train the architecture using the states in the training memory and applying back propagation\cite{speech2013Graves}. 
		\State \quad \quad Derive the average MSE between the states in the test memory and the output of LSTM autoencoder.
		\State \quad \quad If the average MSE is smaller than the minimum MSE so far, set $ k_c^\star = k_c $ and $ k_h^\star = k_h $.
		\State \quad \quad Set $ k_h = k_h+1 $.
		\State \quad \textbf{Until} $ k_h+1 = k^{\textrm{max}}_{h} $.
		\State \textbf{Until} $ k_c+1 = k^{\textrm{max}}_{c} $.
		\State \textbf{Output} $ k^\star_c $ and $ k^\star_h $.
	\end{algorithmic}
	\label{Algorithm:LSTMAutoencoder}
\end{algorithm}
\section{Simulation Results}\label{sec:results}
For our simulations we consider a rectangular area within the following coordinates: $(0,0)$, $(0,1000)$, $(1000,0)$, and $(1000,1000)$. Unless otherwise stated, we consider $ B = 1$ MHz, $ S = 10 $ Mbits, $ \sigma^2 = -100 $ dBm, $ h = 80 $ meters, $v^{\textrm{max}}_x = v^{\textrm{max}}_y = 25$ m/s, and $ \tau = 900 $ seconds. We randomly generate the $ x$ and $y$ coordinates of the initial and final location of UAV as well as the location of the nodes using a uniform distribution on interval $[0,1000]$ meters. Also, the nodes' battery levels are drawn uniformly between 0.1 and 1 joules and the each node;s importance value is drawn uniformly between 0 and 1 and then normalized over the sum of the importance values. We train the UAV, using the ANN architecture in \cite{mnih2015human} with no convolutional neural networks and only one FC layer. We use the Tensorflow-Agents library \cite{TFAgents} for designing the environment, policy, and costs. In addition, we use 8 NVIDIA P100 GPU and 20 Gigabits of memory to train the UAV. All statistical NWAoI results are averaged over 1000 episodes.
\begin{figure}[t!]
    \centering
    \includegraphics[width=0.55\columnwidth]{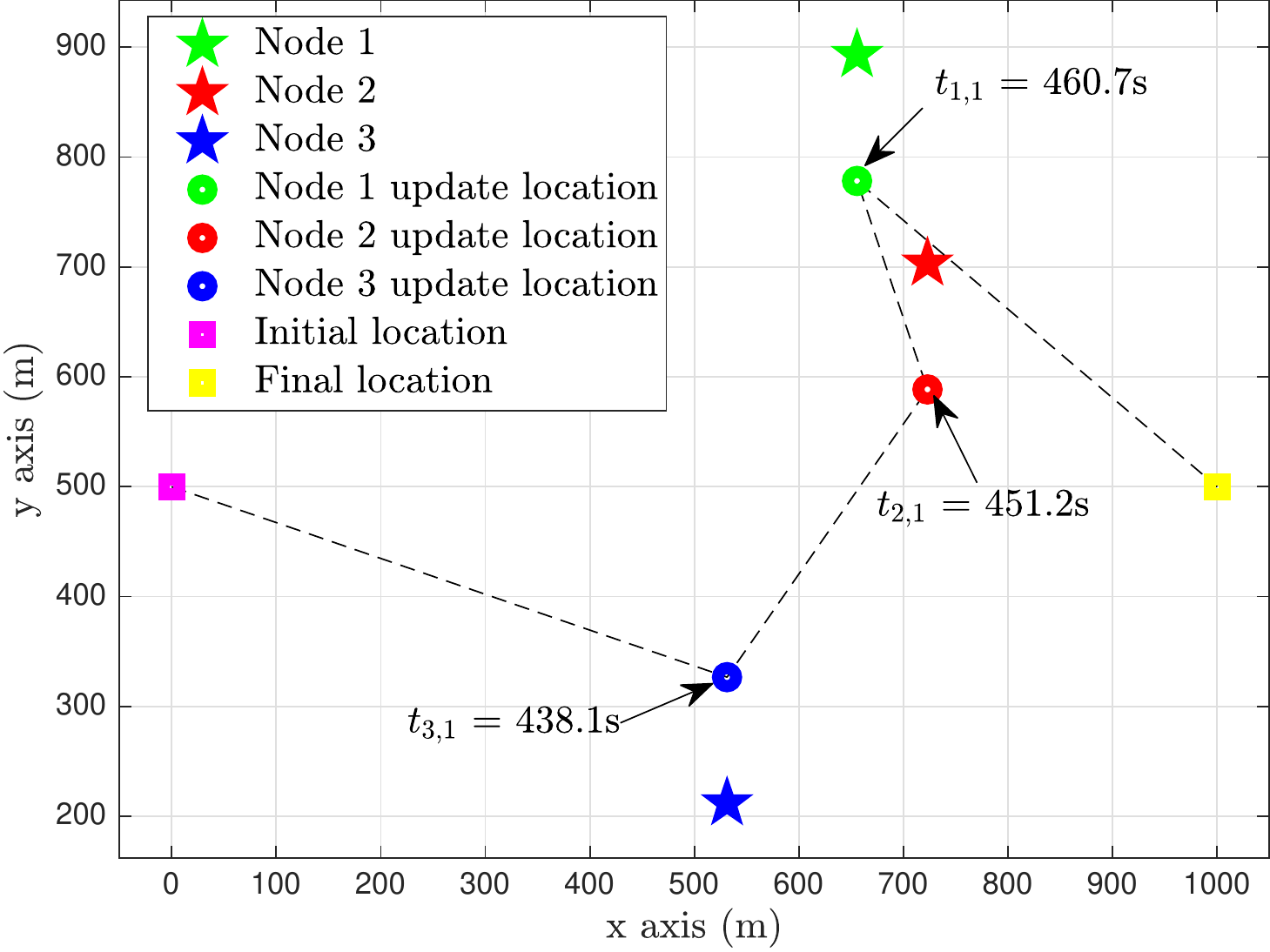}
    \vspace{-3mm}
    \caption{Trajectory optimization using convex optimization.}
    \label{fig:trajectory}
\end{figure}

\subsection{Convex Optimization-based Trajectory}
In Fig. \ref{fig:trajectory}, we consider 3 nodes whose energy levels are randomly drawn between 0.1 and 0.2 joules. The initial and final locations of the UAV are at $ (0,500) $ and $(500,500)$ meters. To study this scenario, we consider a brute force method and solve problem \eqref{eq:objrelaxed} for all of the combinations in \eqref{eq:maxupdate}. Fig. \ref{fig:trajectory} shows the optimal trajectory of the UAV as well as each node's update time instant. Fig. \ref{fig:trajectory} shows that each node can be updated only once during the scenario. Therefore, the UAV tries to update the nodes as close as possible to $ \frac{\tau}{2} = 450$ seconds which is the optimal update time instant when each node  can be updated only once due to Theorem \ref{theorem:lowerbound}. Moreover, Fig. \ref{fig:trajectory} shows that, at the update time instants, the UAV tries to be as close as possible to the nodes in order to consume the least energy.
\begin{figure}
    \centering
    \includegraphics[width = 0.8\textwidth]{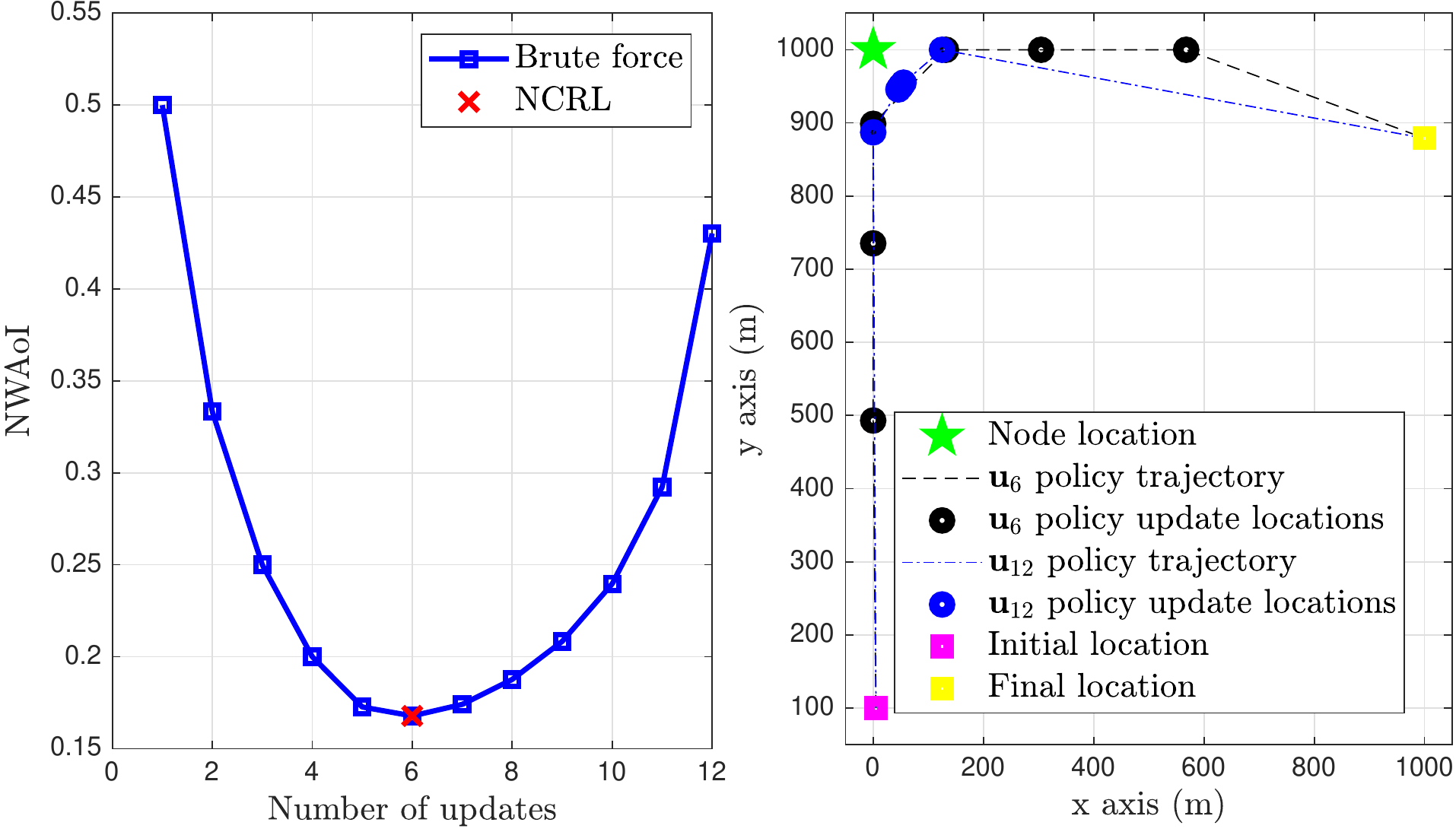}
    \vspace{-3mm}
    \caption{NWAoI vs number of updates for a single node scenario.}
    \label{fig:bruteforce}
\end{figure}

Fig. \ref{fig:bruteforce} presents the impact of the number of updates on NWAoI for a simple scenario with only 1 node that has 1 joule energy. From Fig. \ref{fig:bruteforce}, we observe that the maximum number of times that the UAV can update the node is 12. However, the minimum NWAoI is achieved after 6 updates. This is due to the fact that, as seen in Fig. \ref{fig:bruteforce}, having more updates restricts the node to use small energy levels for each transmission. Therefore, the UAV needs to be closer to the node at each update. This can be seen by comparing the policies with 6 ($\boldsymbol{u}_6$) and 12 ($\boldsymbol{u}_{12}$) updates in Fig. \ref{fig:bruteforce}. Clearly, the update locations of $\boldsymbol{u}_6$ are more uniformly distributed on the UAV's trajectory compared to the update locations of $\boldsymbol{u}_{12}$ which are distributed closely to the node's location. Thus, the difference between the time instants of policy $\boldsymbol{u}_6$ are larger and its resulting NWAoI is smaller than  $\boldsymbol{u}_{12}$. This showcases the importance of action $ a = 0$ which is the terminal action in the optimal policy (since adding more updates for a node does not necessarily reduce the NWAoI). In fact, in Fig. \ref{fig:bruteforce}, we compare the brute force method to the proposed NCRL and show that NCRL can find the optimal number of updates for this scenario.
\begin{figure}
    \centering
    \includegraphics[width = 0.55\textwidth]{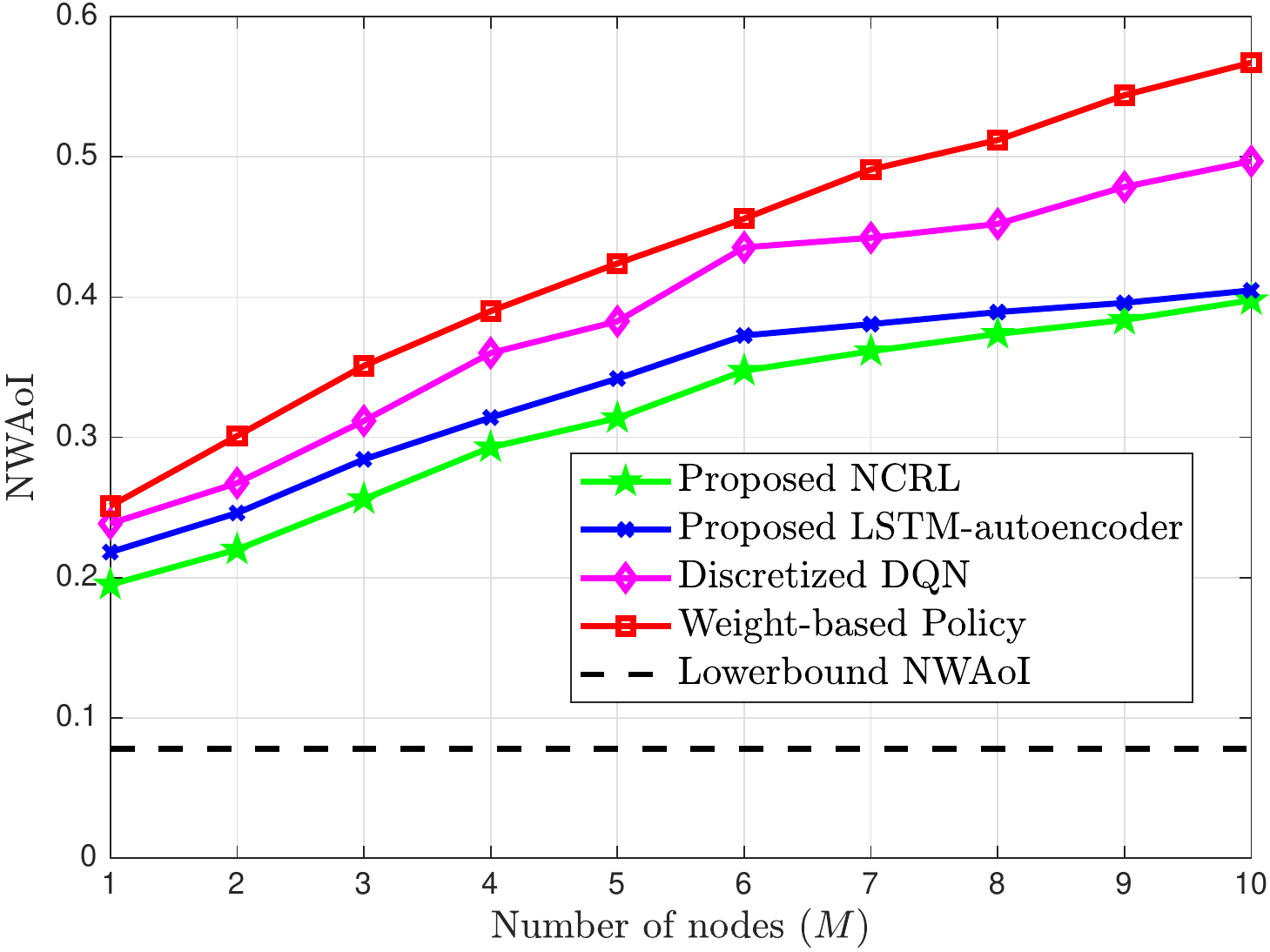}
    \vspace{-3mm}
    \caption{Comparison between the proposed NCRL and LSTM autoencoder with discretized DQN and weight-based policy for small numbers of nodes.}
    \label{fig:nodes_small}
\end{figure}
\subsection{Learning-based scheduling policy}
Fig. \ref{fig:nodes_small} shows the impact of the number of nodes on the NWAoI. In particular, we compare our proposed NCRL and LSTM autoencoder with a discretized DQN approach proposed in our early work \cite{UAV_AoI_DRL1} and the weight-based policy. Fig. \ref{fig:nodes_small} shows that both the proposed methods yield lower NWAoI compared to the discretized DQN and weight-based policies. As the number of nodes increases, the NWAoI increases for all four policies since: 1) each node will update its process less often than the case with smaller network, 2) the action space increases, i.e., the number of feasible scheduling policies increases progressively as shown in \eqref{eq:maxupdate}, which makes finding the optimal policy more challenging, and 3) the spatio-temporal interdependence between the nodes' locations and their update time instants increases. However, as seen from Fig. \ref{fig:nodes_small}, while the gap between NCRL and the discretized/weight-based policy increases when the number of nodes increases, the gap between NCRL and LSTM autoencoder reduces. This is because, for larger number of nodes, the LSTM autoencoder starts showing its impact in learning the spatio-temporal interdependence between the states of the problem. From Fig. \ref{fig:nodes_small}, we can also observe that the lower bound on NWAoI (expressed in \eqref{eq:minNWAoI}) does not depend on the number of nodes since it is only a function of the nodes' weights and maximum number of allowable updates, i.e., the number of nodes does not have any impact on that lower bound value.
\begin{figure}
    \centering
    \includegraphics[width = 0.55\textwidth]{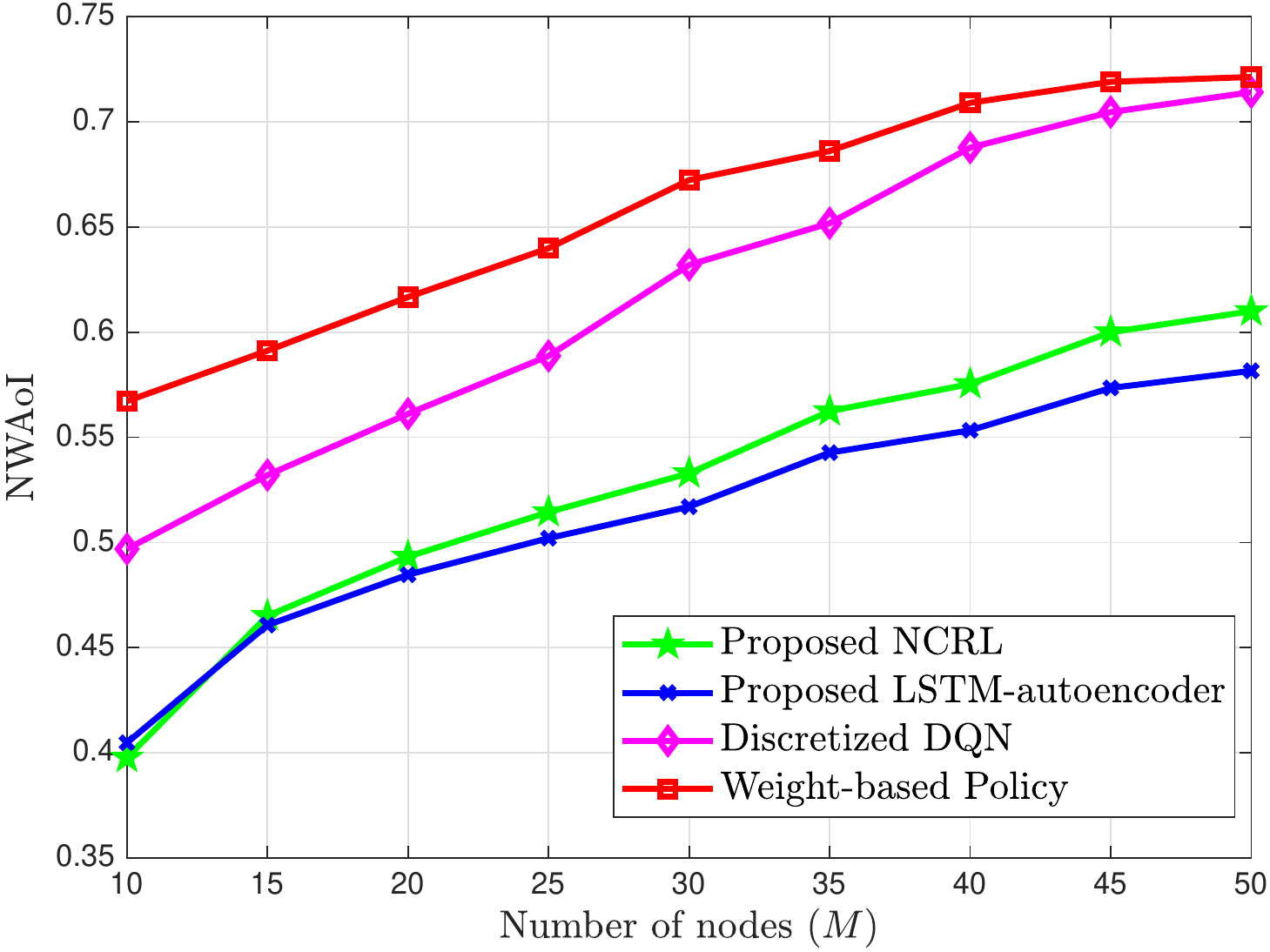}
    \vspace{-3mm}
    \caption{Comparison between the proposed NCRL and LSTM autoencoder with discretized DQN and weight-based policy for large numbers of nodes.}
    \label{fig:nodes_large}
\end{figure}

In Fig. \ref{fig:nodes_large}, we study the impact of having a large number of nodes on the performance of the four policies. Fig. \ref{fig:nodes_large} demonstrates that, as the number of nodes increases, the proposed LSTM autoencoder shows its impact and results in a smaller NWAoI compared to NCRL. This shows that the proposed LSTM autoencoder can capture some interdependencies between the states of the problem that only using the last column of the state will fail to capture. Therefore, an LSTM autoencoder can learn better policies compared to NCRL. The reason why the LSTM autoencoder cannot outperform NCRL for a small number of nodes is because its accuracy is not 100\% when finding the state representation for short sequence sizes. Therefore, for a small network, the test error prevents LSTM autoencoder to outperform NCRL. However, for a large network of nodes, the benefits of using the LSTM autoencoder is larger than its test error, and, thus, it can outperform NCRL. Fig. \ref{fig:nodes_large} also shows that for large-scale networks, the discretized DQN in \cite{UAV_AoI_DRL1} fails to even outperform the weight-based policy since, in this method, the state space grows exponentially which makes it harder for the DQN to learn a good policy.
\begin{figure}
    \centering
    \includegraphics[width = 0.55\textwidth]{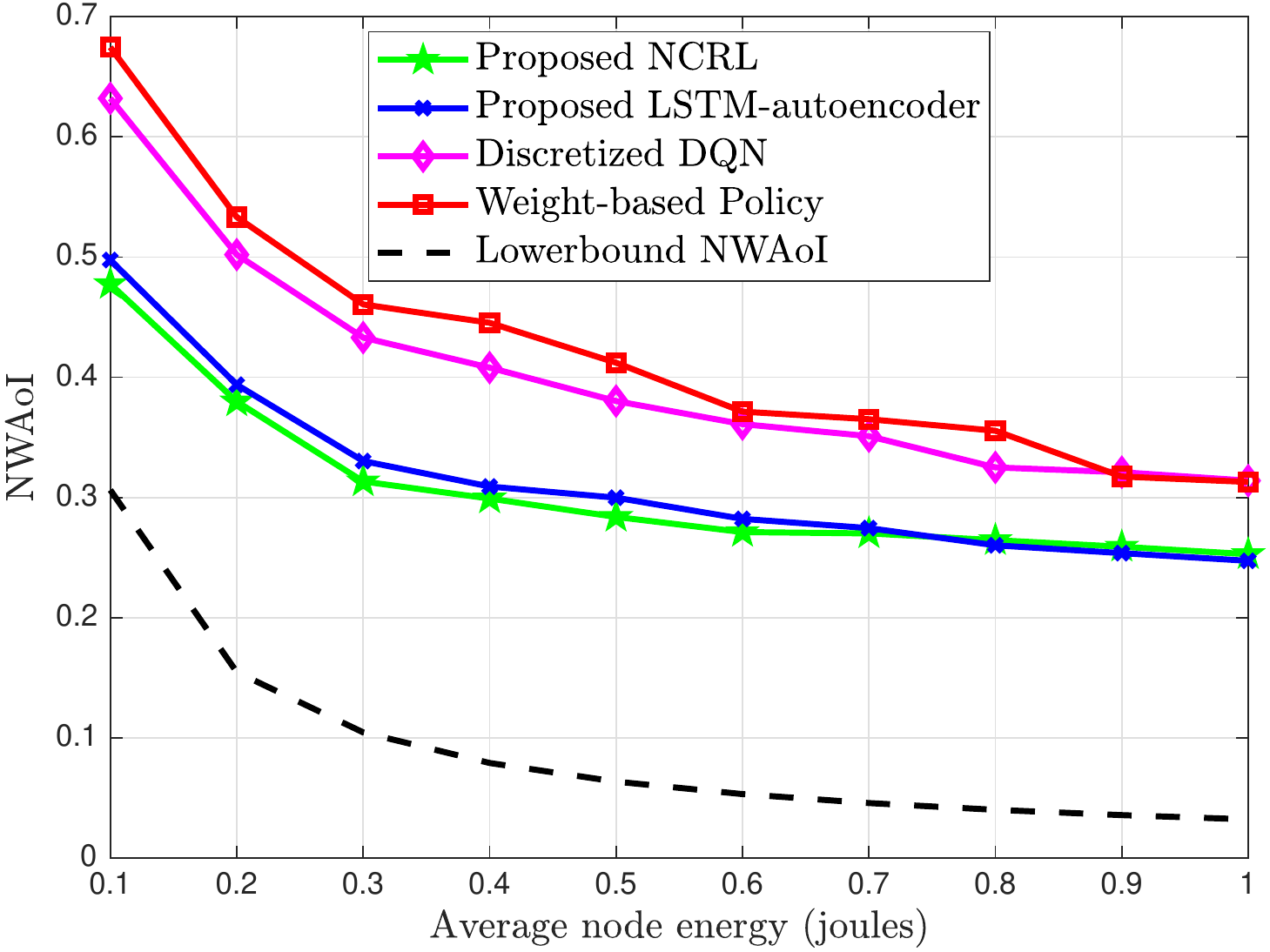}
    \vspace{-3mm}
    \caption{The impact of node energy levels on NWAoI.}
    \label{fig:energy}
\end{figure}

Fig. \ref{fig:energy} shows the impact of the node energy level on NWAoI. In Fig. \ref{fig:energy}, we consider 3 nodes with energy levels randomly drawn from: $\left[0.05,0.15\right]$, $\left[0.15,0.25\right]$, $\dots$, $\left[0.95,1.05\right]$ joules. Thus, the average energy level of the nodes will be between $0.1$ and $1$ joules. Fig. \ref{fig:energy} demonstrates that the proposed NCRL and LSTM autoencoder can outperform the discretized DQN and weight-based policies. Moreover, Fig. \ref{fig:energy} shows that as the energy level of the nodes increases, the LSTM autoencoder achieves lower NWAoI compared to NCRL. This is due to the fact that larger energy levels help the UAV update the nodes for a larger number of times which, in turn, increases the size of the state matrix $\boldsymbol{S}_{n}$. Therefore, the effect of the LSTM autoencoder can be more obvious when the nodes' energy levels increase. From Fig. \ref{fig:energy}, we also observe that for a larger average node energy, the discretized DQN cannot achieve a good performance and in some cases the weight-based policy has a lower NWAoI. This is because of the nature of the discretized DQN approach in \cite{UAV_AoI_DRL1} where the state space and the complexity of the problem grow progressively with the energy levels of the nodes while the weight-based policy's complexity does not depend on the energy levels. Fig. \ref{fig:energy} also shows that the lower bound on NWAoI decreases sub-linearly with respect to the average node energy level which means that the impact of energy level reduces gradually as the node energy levels increase. Such a sub-linear behavior can be noticed also for all four policies.
\begin{figure}
    \centering
    \includegraphics[width = 0.55\textwidth]{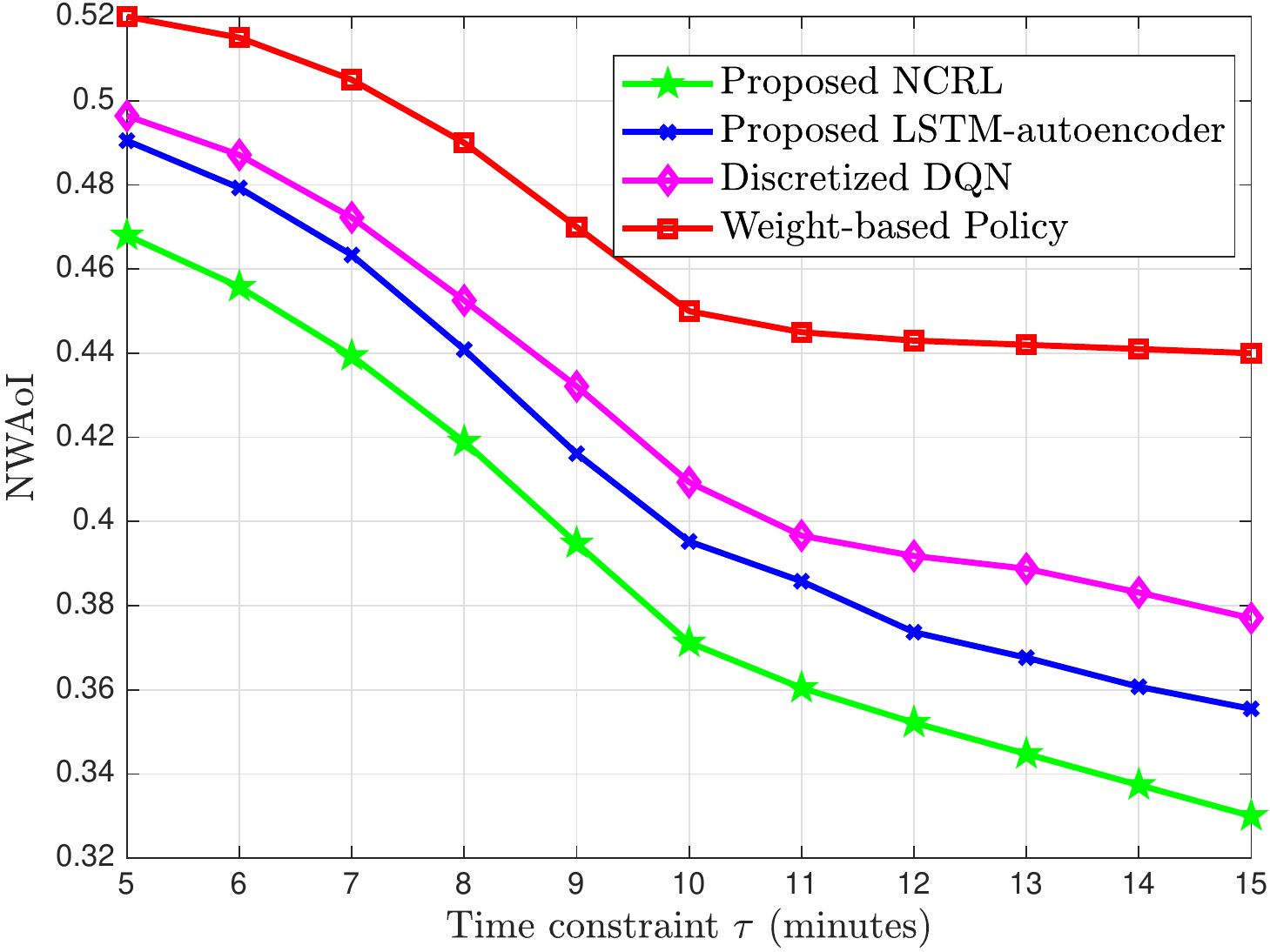}
    \vspace{-3mm}
    \caption{The impact of time constraint on NWAoI.}
    \label{fig:timeconstraint}
\end{figure}

Fig. \ref{fig:timeconstraint} compares the performance of the proposed NCRL and LSTM autoencoder with discretized DQN and weight-based policies as a function on the time constraint $\tau$. We consider 3 nodes and solve the problem for different scenarios with time constraint between 5 to 15 minutes. Fig. \ref{fig:timeconstraint} shows that, as the time constraint increases, the NWAoI becomes smaller since a larger time constraint gives more opportunity to the UAV to move closer to the nodes and update the node status more frequently. Moreover, Fig. \ref{fig:timeconstraint} shows that the proposed NCRL and LSTM autoencoder can outperform the discretized DQN and weight-based policies. Furthermore, the performance gap between the four policies stay fixed which indicates that the time constraint has a general impact on the solution of the problem and does not depend on the policy type.
\begin{figure}
    \centering
    \includegraphics[width = 0.55\textwidth]{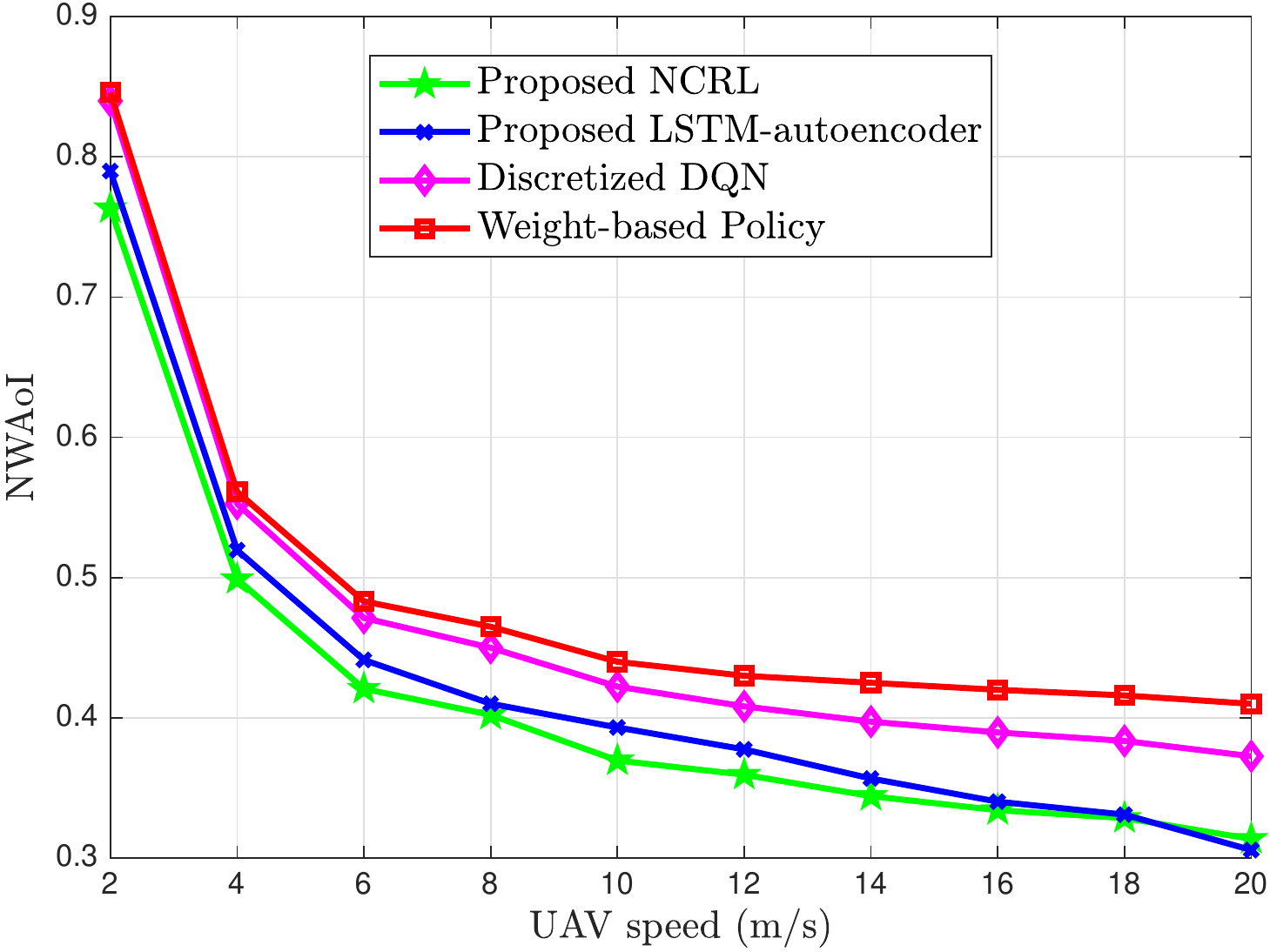}
    \vspace{-3mm}
    \caption{The impact of the UAV speed on NWAoI.}
    \label{fig:speed}
\end{figure}

In Fig. \ref{fig:speed}, we consider three nodes while the UAV speed varies between 2 and 20 m/s. From Fig. \ref{fig:speed} we notice that for small values of UAV speed, NWAoI is almost similar for NCRL, LSTM autoencoder, discretized DQN, and weight-based policy since the UAV cannot cover large areas and due to its time constraint it may not even update any node. However, as the UAV speed increases, the NWAoI also decreases because the UAV can move around faster and can update nodes more frequently. Fig. \ref{fig:speed} demonstrates that LSTM autoencoder can achieve even lower NWAoI values compared to NCRL for higher UAV speeds. This is due to the fact that, the number of updates increases with the increase in speed which results in larger state matrices. Therefore, the LSTM autoencoder can learn a better representation of the state which will result in learning better policies.

\section{Conclusion}\label{sec:con}
In this paper, we have investigated the problem of minimizing the NWAoI for a UAV-assisted wireless network in which a UAV collects status update packets from energy-constrained ground nodes. First, we have formulated the problem as a mixed-integer program. Then, for a given scheduling policy, we have proposed a convex optimization-based approach to
obtain the UAV's optimal flight trajectory and time instants on updates. However, due to the combinatorial nature of the formulated problem, it is very challenging to find the optimal scheduling policy. To overcome this hurdle, we have proposed a novel NCRL algorithm using DQN to reduce the system state complexity while learning the optimal scheduling policy at the same time. However, for large-scale networks, the DQN cannot efficiently learn the optimal scheduling policy. Therefore, we have then proposed an LSTM autoencoder that can help the proposed deep RL to learn a better policy for such large-scale scenarios. We have analytically derived a lower bound on the minimum NWAoI, and obtained an upper bound on the UAV's minimum speed to achieve that lower bound value. Our numerical results have shown that the proposed NCRL algorithm significantly outperforms baseline policies, such as the discretized DQN and weight-based policies, in terms of the achievable NWAoI per process. They have also demonstrated that the achievable NWAoI by the proposed algorithm is monotonically decreasing with the time constraint of the UAV, the battery sizes of the ground nodes, and the UAV speed.
\begin{appendix}
\subsection{Proof of Theorem \ref{theorem:lowerbound}}\label{App:Theorem1}
	The minimum required energy for an update from a node $ m $ is $ \frac{\sigma ^2 \left(2^{S/B} - 1\right)h^2}{\beta_0} $ which is the case when the UAV requests for update from node $ m $ while it stays on top of node $ m $, i.e. $ x_{m,i} = x_m $ and $ y_{m,i} = y_i $. In this case, every node $ m $, will be updated $ \bar{n}_m $ times in the entire $ \tau $ seconds. However, this requires UAV to move from the top of a node to top of another node in less than the time difference between two optimal consecutive update time instants. Therefore, in order to find the lower bound on NWAoI, we neglect the limit on the UAV's speed and find the optimal update time instants for each node. Note that, in this case, we assume that \eqref{eq:xspeedconst} and \eqref{eq:yspeedconst} are always satisfied. Here, we define $ \delta_{i,m} \triangleq t_{i,m} - t_{i-1,m}$ as the difference between two update time instants of node $ m $. Then, we have:
	\begin{align}\label{eq:objdelta}
	\bar{G} = \frac{1}{\tau^2}\sum_{m=1}^{M}\lambda_{m}\left(\sum_{i = 1}^{\bar{n}_m}\delta_{i,m}^2 + \left(\tau - \sum_{i = 1}^{\bar{n}_m}\delta_{i,m}\right)^2\right).
	\end{align}
	Since \eqref{eq:objdelta} is a convex function, we take the first derivative of $\bar{G}$ with respect to $ \delta_{i,m} $, for $1 \leq m \leq M $ and $ 1 \leq i \leq \bar{n}_{m} $, and set it equal to 0 in order to find the optimal update time instants which yields:
	\begin{align}
		\frac{\partial \bar{G}}{\partial \delta_{i,m}} = \frac{2\lambda_{m}}{\tau^2}\left(\delta_{i,m} - \left(\tau - \sum_{j = 1}^{n_m}\delta_{j,m}\right)\right) = \frac{2\lambda_{m}}{\tau^2} \left(2\delta_{i,m} + \sum_{j = 1, j \neq i}^{n_m}\delta_{j,m} - \tau\right).
	\end{align}
	Thus, for every node $ m $ the optimal values for $ \delta_{i,m} $ is the solution of the following equation:
	\begin{align}\label{eq:matrix}
		\left[\begin{array}{c c c c}
		2 & 1 & \cdots & 1\\
		1 & 2 & \ddots & \vdots\\
		\vdots & \ddots & \ddots & 1 \\
		1 & \cdots & 1 & 2
		\end{array}\right]\left[\begin{array}{c}
		\delta_{1,m}\\\vdots\\\delta_{\bar{n}_m,m}
		\end{array}\right] = \left[\begin{array}{c}
		\tau\\\vdots\\\tau
		\end{array}\right].
	\end{align}
	Now, if we subtract the first row of the matrix in \eqref{eq:matrix} from all of the other rows we will have:
	\begin{align}\label{eq:matrix2}
		\left[\begin{array}{c c c c c}
		2 & 1 & \cdots & \cdots & 1\\
		-1 & 1 & 0 & \cdots & 0 \\
		\vdots & 0 & \ddots & \ddots & \vdots\\
		\vdots & \vdots & \ddots &\ddots & 0\\
		-1 & 0 & \cdots & 0 & 1
		\end{array}\right]\left[\begin{array}{c}
		\delta_{1,m}\\\vdots\\\delta_{\bar{n}_m,m}
		\end{array}\right] = \left[\begin{array}{c}
		\tau\\ 0 \\ \vdots\\0
		\end{array}\right],
	\end{align} 
	which yields  $ \delta_{1,m} = \delta_{2,m} = \dots = \delta_{\bar{n}_m,m} = \frac{\tau}{\bar{n}_m + 1} $. Therefore, the optimal NWAoI will be:
	\begin{align}
		\bar{G}_{\textrm{min}} = \frac{1}{\tau^2}\sum_{m=1}^{M}\lambda_{m}\left(\bar{n}_m + 1\right)\left(\frac{\tau}{\bar{n}_m + 1}\right)^2,
	\end{align}
	which can be simplified to \eqref{eq:minNWAoI}.
\end{appendix}
\bibliographystyle{IEEEtran}
\bibliography{OptAoI}
\end{document}